\documentclass[11pt]{article}

\usepackage{times}
\usepackage{color}
\usepackage{fullpage}
\usepackage{amsthm}
\usepackage[ruled, vlined, linesnumbered]{algorithm2e}
\usepackage{color}
\usepackage{xcolor}
\usepackage{array}
\usepackage{diagbox}
\usepackage{multirow}
\usepackage{amsmath,amssymb}
\usepackage[colorlinks,linkcolor=blue,citecolor=blue,urlcolor=black]{hyperref}
\usepackage{caption, subcaption}
\DeclareCaptionType{copyrightbox}
\usepackage{framed,url}
\usepackage{booktabs}
\usepackage{tabularx}
\usepackage{mathrsfs}

\newtheorem{theorem}{Theorem}
\newtheorem{lemma}[theorem]{Lemma}

\newtheorem{definition}{Definition}

\newcommand{\abs}[1]{{\left | #1 \right |}}
\newcommand{\gray}[1]{\textcolor{gray}{#1}}
\newcommand{\E}{\mathbf{E}}
\newcommand{\Var}{\mathbf{Var}}
\newcommand{\eps}{\epsilon}

\newcommand{\A}{\mathcal{A}}
\renewcommand{\Pr}{\mathbf{Pr}}
\newcommand{\cE}{\mathcal{E}}

\newcommand{\DI}{{\text{DI}}}

\newcommand{\INDEX}{{\sc INDEX}}
\newcommand{\DISJ}{{\sc DISJ}}

\newcommand{\EXP}{\operatorname{Exp}}

\SetKw{Break}{break}
\SetKwData{Null}{null}

\newcommand{\rFo}{\DI}
\newcommand{\rLo}{L_{\DI}}
\newcommand{\rFp}{M_p}

\begin{document}

\title{Statistics of Similarity Graphs in Node-Arrival Streams}

\author{
        Kaiwen Liu \\ 
        Computer Science Department\\
        Indiana University\\
        \texttt{kaiwliu@iu.edu}
        \and 
        Qin Zhang \\
        Computer Science Department\\
        Indiana University\\
        \texttt{qzhangcs@iu.edu}
    }

\maketitle

\begin{abstract}
In this paper, we study several statistical problems on similarity graphs in the node-arrival streaming model, including degree moments, diversity index, degree-moment sampling, and diversity sampling. We develop constant-pass, sublinear-space streaming algorithms for these problems and establish space lower bounds that nearly match the upper bounds in their dependence on the stream length.
\end{abstract}

\section{Introduction}
\label{sec:intro}

Many statistical problems in data stream analysis, including distinct elements and frequency moments, are defined in terms of item frequencies, which count how many times each item occurs. In applications involving data provided to or generated by large language models, however, this notion may be too restrictive. For example, different prompts may request essentially the same task using different wording. Similarly, an LLM may generate documents that appear different but are semantically equivalent because of paraphrasing or summarization. In such settings, when measuring how often an item occurs, similar but nonidentical items should also be allowed to contribute. This motivates a similarity-aware generalization of frequency, in which each stream item contributes according to its similarity to the item of interest.

Once each pair of stream items is assigned a similarity score, the stream naturally induces a similarity graph, with stream items as nodes and edge weights representing pairwise similarities. The weighted degree of a
node is precisely its similarity-aware frequency. Moreover, the graph also captures the relational structure that is not revealed in individual frequencies. 

In our setting, the similarity graph is implicit: the stream presents only
the data items, rather than the edges or edge weights of the graph, and
pairwise similarities can be obtained only by evaluating a similarity
function on items stored in memory. This leads to the central question of
this paper:
\begin{quote}
    \emph{How space-efficiently can we perform fundamental statistical tasks on an implicit similarity graph presented as a node-arrival stream?}
\end{quote}

We formalize this setting as follows.
Let $\sigma = (\sigma_1, \ldots, \sigma_n) \in U^n$ denote the data stream, and let $f : U \times U \to [0,1]$ be a symmetric similarity function satisfying $f(x, x) = 1$ for all $x \in U$. 
The similarity graph of $\sigma$ is the complete weighted graph $G^\sigma$ on the node set $[n]$, with edge weight function $w(i,j) = f(\sigma_i, \sigma_j)$ for all $i, j\in [n]$.
Thus, each edge has a weight in $[0,1]$, with every node $i$ having a self-loop of weight $w(i,i)=1$.
For each $i \in [n]$, let
\begin{equation*}
    d_i = \sum_{j \in [n]} w(i,j) = \sum_{j \in [n]} f(\sigma_i, \sigma_j),
\end{equation*}
which can be viewed as the similarity-weighted frequency of $\sigma_i$.

In this paper, we study the following statistics of the similarity graph:
\begin{enumerate}
    \item $M_p(\sigma) = \sum_{i \in [n]} d_i^p\ (p > 0)$: {\em degree moments} of $G^\sigma$. When $f$ is the equality function, $M_p(\sigma)$ degenerates to the classical $(p+1)$-th frequency moment $F_{p+1}(\sigma)$. When $f$ is Boolean, the quantity $M_1(\sigma)$ corresponds to the size of the self-similarity join of the set $\sigma$ under $f$.
    \smallskip
    
    \item $\DI(\sigma) = \sum_{i \in [n]} d_i^{-1}$ : {\em diversity index} of $G^\sigma$. $\DI(\sigma)$ characterizes the ``level of diversity'' of  the set $\sigma$ under the similarity function $f$. When $f$ is the equality function, $\DI(\sigma)$ degenerates to the number of distinct elements in $\sigma$.
    When $f$ is Boolean, 
    this quantity is also known as the Caro–Wei bound~\cite{Caro79,Wei81} in the literature, and it gives a lower bound on the size of a maximum independent set in the graph.
        \smallskip

    \item $L_{M_p}$-sampler: {\em degree-moment sampler} on $G^\sigma$ that samples nodes with probability proportional to the {\em $p$-th power} of their degrees; it outputs node $i$ with probability $d_i^p / M_p(\sigma)$. When $f$ is the equality function, $L_{M_p}$-sampler degenerates to the classical $\ell_{p+1}$-sampler.
    \smallskip
    
    \item $L_{\DI}$-sampler: {\em diversity sampler} on $G^\sigma$ that samples nodes with probability proportional to the {\em inverse} of their degrees; it outputs node $i$ with probability $d_i^{-1} / \DI(\sigma)$. When $f$ is the equality function, $L_{\DI}$-sampler degenerates to the classical $\ell_0$-sampler.
\end{enumerate}

\vspace{2mm}
\noindent{\bf The Computational Model.\ }
We work in the well-established data stream model~\cite{FM85,AMS99} for large-scale data processing. 
In this model, a sequence of data items arrive sequentially, while the available memory is too small to store the entire dataset. 
Our goal is to compute a function of the dataset using one or more linear scans of the stream. 
The main objective is to minimize the space usage of the algorithm and the number of passes needed over the data.  

In this paper, we focus on \emph{node-arrival streams}. The input is a sequence of items
$\sigma=(\sigma_1,\ldots,\sigma_n)\in U^n$, where each item corresponds to a node of
$G^\sigma$. We are also given access to a similarity function
$f:U\times U\to[0,1]$. The weights of $G^\sigma$ are \emph{not} provided explicitly.
Instead, the weight between two nodes is computed on demand by evaluating $f$, but only when the items corresponding to both nodes are stored in memory.

We believe that computing statistics of an implicit similarity graph provides a natural setting for studying node-arrival streams. In this context, nodes (data items) are the primary objects, while pairwise weights are derived objects that capture relationships or similarities between the nodes.

\vspace{2mm}
\noindent{\bf Notation and Conventions.\ }
We assume $p$ is a constant throughout the paper. When $\sigma$ is clear from the context, we simply write $M_p(\sigma)$, $\DI(\sigma)$, $L_{M_p}(\sigma)$ and $L_\DI(\sigma)$ as $M_p$, $\DI$, $L_{M_p}$ and $L_\DI$.  

We measure the space complexity of our upper bounds in words and use a unit-cost RAM model. We assume that each stream item, each value returned by the similarity function $f$, and each arithmetic quantity maintained by our algorithms can be represented using $O(1)$ words. If each item requires $b$ bits to represent, then all of our upper bounds, when stated in bits, incur an additional multiplicative factor of $b$. Our lower bounds are stated in bits and use only Boolean similarity functions.

We regard each stream item as a distinct node. Thus, node $i\in[n]$ is represented by the pair $(i,\sigma_i)$. When there is no ambiguity, we use $\sigma_i$ to refer to node $i$ as well as its associated item value. 

For any $\eps \in (0,1)$, we say  $\tilde{X}$ is a $(1+\eps, \delta)$-approximation of $X$ if $(1-\eps)X \le \tilde{X} \le (1+\eps)X$ with probability at least $(1 - \delta)$. For some lower bounds, we are able to prove for larger approximation ratios. For any $\alpha \ge 1$, we say $\tilde{X}$ is an $(\alpha, \delta)$-approximation of $X$ if $X/\alpha \le \tilde{X} \le \alpha X$ with probability at least $(1 - \delta)$.  

We use the notation $x = a \pm b$ to indicate that $x$ lies within the interval $[a - b, a + b]$.
For $L_\DI$-sampling and $L_{M_p}$-sampling, we slightly relax the sampling criterion by allowing a relative error $\eps$.  A $(1 + \eps, \delta)$-$L_\DI$-sampler (resp. $L_{M_p}$-sampler) on $G^\sigma$ fails with probability at most $\delta$ and conditioned on no failure, outputs node $i$ with probability $p_i = (1 \pm \eps) \frac{d_i^{-1}}{\DI(\sigma)}$
(resp. $p_i = (1 \pm \eps) \frac{d_i^p}{M_p(\sigma)}$).
 
\vspace{2mm}
\noindent{\bf Our Results.\ }
We present the following algorithmic results in Section~\ref{sec:algo}.
\begin{itemize}
    \item For any $\eps \in (0, 1)$, a $(1+\eps, o(1))$-approximation algorithm for  diversity index $\DI$ uses $O\left(\frac{1}{\eps^2}{\sqrt{n}\log n}\right)$ words of space and three passes.

    \item  For any $\eps \in (0, 1)$, a $\left(1+\eps, o(1)\right)$-$L_{\DI}$-sampler uses $O\left(\frac{1}{\eps^2}{\sqrt{n}\log n}\right)$ words of space and three passes.
    
    \item  For any $\eps \in (0, 1)$ and any $p > 0$, a $(1+\eps, 0.01)$-approximation algorithm for degree moment $M_p$ uses $O\left(\frac{1}{\eps^2}{n^{1-\frac{1}{p+1}}\log n}\right)$ words of space and two passes.

    \item For any $\eps \in (0, 1)$ and any $p > 0$, a $\left(1+\eps, o(1)\right)$-$L_{M_p}$-sampler uses $O\left(\frac{1}{\eps^2}{n^{1-\frac{1}{p+1}}\log n}\right)$ words of space and two passes.
\end{itemize}

We supplement the results above with the following lower bounds in Section~\ref{sec:lb}.
\begin{itemize}
    \item  Any one-pass $(O(1), 0.49)$-approximation algorithm for diversity index \DI\ needs $\Omega(n)$ bits of space. For any $\eps \in (0, 0.9)$, any $O(1)$-pass $(1+\eps, 0.49)$-approximation algorithm for \DI\ needs $\Omega\left(\min\left(\sqrt{\frac{n}{\eps}}, n\right)\right)$ bits of space. 

    \item For any $\eps \in (0, 0.9)$, any one-pass $(1+\eps,0.49)$-$L_{\DI}$-sampler needs $\Omega(n)$ bits of space, and any $O(1)$-pass $(1+\eps,0.49)$-$L_{\DI}$-sampler needs $\Omega(\sqrt{n})$ bits of space.
    
    \item For any $p > 0$, any one-pass $(O(1), 0.49)$-approximation algorithm for $M_p$ needs $\Omega(n)$ bits of space. For any $\eps \in (0, 0.9)$ and $p > 0$, any $O(1)$-pass $(1+\eps, 0.49)$-approximation algorithm for $M_p$ needs $\Omega\left(\min \left(\eps^{-\frac{1}{p+1}} n^{1-\frac{1}{p+1}}, n\right)\right)$ bits of space.

    \item For any $\eps \in (0, 0.9)$ and $p > 0$, any one-pass $(1+\eps,0.49)$-$L_{M_p}$-sampler needs $\Omega(n)$ bits of space, and any $O(1)$-pass $(1+\eps,0.49)$-$L_{M_p}$-sampler needs $\Omega\left(n^{1-\frac{1}{p+1}}\right)$ bits of space.
\end{itemize}

We summarize our main results in Table~\ref{tab:summary}. The one-pass lower bounds rule out the possibility of obtaining one-pass sublinear-space streaming algorithms for diversity index, diversity sampling, degree moments and degree-monent sampling. Our constant-pass upper and lower bounds match in their dependence on $n$, up to logarithmic factors.

\begin{table}[t]
    \centering
    \setlength{\tabcolsep}{3pt}
    \renewcommand{\arraystretch}{1.18}

    \begin{tabularx}{\textwidth}{@{}
        >{\raggedright\arraybackslash}p{0.15\textwidth}
        >{\raggedright\arraybackslash}p{0.27\textwidth}
        >{\raggedright\arraybackslash}p{0.30\textwidth}
        >{\raggedright\arraybackslash}X
        @{}}
        \toprule
        Problem
        &
        Upper bound
        &
        Constant-pass lower bound
        &
        One-pass lower bound
        \\
        \midrule

        \DI
        &
        $3$ passes;
        $\widetilde{O}(\epsilon^{-2}\sqrt{n})$
        &
        $\Omega\!\left(\min\{n,\sqrt{n/\epsilon}\}\right)$
        &
        $\Omega(n)$
        \\

        \addlinespace

        $L_\DI$-sampling
        &
        $3$ passes;
        $\widetilde{O}(\epsilon^{-2}\sqrt{n})$
        &
        $\Omega(\sqrt{n})$
        &
        $\Omega(n)$
        \\

        \addlinespace

        $M_p\ (p > 0)$
        &
        $2$ passes;
        $\widetilde{O}\!\left(\epsilon^{-2}n^{1-\frac{1}{p+1}}\right)$
        &
        $\Omega\!\left(\min\left\{n,\epsilon^{-\frac{1}{p+1}}n^{1-\frac{1}{p+1}}\right\}\right)$
        &
        $\Omega(n)$
        \\

        \addlinespace

        $L_{M_p}$-sampling
        &
        $2$ passes;
        $\widetilde{O}\left(\epsilon^{-2}n^{1-\frac{1}{p+1}}\right)$
        &
        $\Omega\left(n^{1-\frac{1}{p+1}}\right)$
        &
        $\Omega(n)$
        \\

        \bottomrule
    \end{tabularx}
     \caption{
        Summary of our results.
        Upper bounds are in words and lower bounds are in bits;
        $\widetilde{O}$ hides logarithmic factors in $n$.   All upper bounds hold for $(1+\eps, 0.01)$-approximation. The approximation guarantees for lower bound vary; see ``Our Results'' for details.
    }
    \label{tab:summary}
\end{table}
\subsection{Related Work}

\noindent{\bf Statistical Analysis in Data Streams.\ }
Many statistical functions have been studied in the data stream literature, including distinct elements ($F_0$)~\cite{AMS99,BJKST02,DF03,ganguly07,BHR+07,KNW10,CVM22}, $\ell_0$-sampling~\cite{FIS05,DCMR05,JST11,CF14}, heavy hitters~\cite{MG82,CCF02,CM05,MAE05}, frequency moments ($F_p$) and $\ell_p$-sampling~\cite{LW13,AMS99,IW05,BGK+06,MW10,AKO11,BO10,Andoni17, Ganguly11, Woodruff04, BYJ+04, CKS03, Ganguly12, BZ25, GW18,KNW10b}, etc.  Most of these algorithms achieve space complexity that is polylogarithmic in the universe size.   As mentioned earlier, taking $f$ to be the equality function, $F_0(\sigma)$ coincides with $\DI(\sigma)$, and for every $p \ge 1$, the frequency moment $F_p(\sigma)$ corresponds to $M_{p-1}(\sigma)$.
However, algorithms developed for classical data streams cannot be applied to the statistical estimation of the similarity graph, where, as we have shown in this paper, strong polynomial space lower bounds hold.

\vspace{2mm}
\noindent{\bf Node-Arrival Streams.\ }
Node-arrival streams differ from the more widely studied edge-arrival model, in which the input consists of the edges of a graph (see, e.g., \cite{McGregor14} for a survey). In the literature, the model considered in this paper is known as the \emph{implicit vertex-arrival model} and has previously been studied in \cite{EHR16, CP17, CDK19}. These works, however, focus specifically on the \emph{maximum independent set} problem for certain geometric intersection graphs. More recently, \cite{LVZ26} studied the cost of correlation clustering in the node-arrival model, whereas our focus is on statistical problems over similarity graphs.

\vspace{2mm}
\noindent{\bf Similarity-Aware Statistical Analysis.\ }
In more recent years,  a few works have investigated data streams with near-duplicates~\cite{CZ16,CZ18,Zhang25}. The work \cite{CZ16} studied the $F_0$ problem in $O(1)$-dimensional Euclidean spaces and in metric spaces that admit efficient locality sensitive hashing.  Its follow-up \cite{CZ18} examined $\ell_0$-sampling in the same settings.  More recently, \cite{Zhang25} generalized the study of $F_0$ to general metric spaces. However, these works focus on a particular class of datasets, referred to as {\em well-shaped} data, corresponding to the special case of our framework
in which $f$ is Boolean and transitive. Although these works also consider the case when transitivity does not hold and introduced the notion of {\em min-partition based ambiguity}, this definition is specifically designed for the $F_0$ problem and does not readily extend to other statistical functions.  

A recent work~\cite{LZ26} introduced a statistical estimation framework for noisy datasets based on the notion of \emph{mismatch-ambiguity}. 
This approach also relies on the similarity graph representation $G^\sigma$ of the input dataset $\sigma$ induced by a similarity function. 
However, rather than directly computing statistics of $G^\sigma$, it assumes the existence of an underlying noiseless ground-truth dataset $\tau$ corresponding to the observed noisy dataset $\sigma$. 
The goal is then to estimate the statistics of $\tau$, with the approximation ratio parameterized by a carefully defined mismatch-ambiguity measure $\eta$, which captures the discrepancy between the similarity graphs $G^\sigma$ and $G^\tau$.
The guarantees provided by this framework are meaningful only when $\eta$ is sufficiently small. 
At the same time, $\eta$ itself is unknown in practice, since the ground-truth dataset $\tau$ is not available.

\vspace{2mm}
\noindent{\bf Other Related Work.\ }
Degree moments $M_p$ have also been studied in the sublinear-time algorithms~\cite{Feige04, GR08, GRS11, ERS17}. That line of work focuses on achieving sublinear running time in the query model and assumes access to multiple types of graph queries (e.g., neighbor queries and degree queries), which differs substantially from our streaming setting.  

The Caro–Wei bound has been studied in the streaming setting~\cite{CDK18}, but only for edge-arrival streams and the explicit vertex-arrival model. In the latter model, when a vertex $v$ arrives, the stream explicitly presents all edges $(u,v)$ such that $u$ has arrived earlier.

\section{Algorithms}
\label{sec:algo}

In this section, we present the algorithms for diversity index, diversity sampler, degree moments and degree-moment sampling. We assume that $n$ is known in advance to simplify the presentation. In Section~\ref{app:assumption}, we discuss how to remove this assumption without an additional pass.

\subsection{Diversity Index}
\label{sec:F0}

\begin{algorithm}[!t]
    \caption{Diversity index estimation}
    \label{alg:F0}
    \DontPrintSemicolon
    \SetAlgoNoEnd
    
    \KwIn{a data stream $\sigma$ of size $n$ and a similarity function $f : U \times U \to [0, 1]$, parameter $\eps$}
    
    \KwOut{a $(1 +\eps)$-approximation of $\rFo(\sigma)$}
    \smallskip
    $t_1 \gets \frac{48\sqrt{n} \cdot \log n}{\eps^2}$, $t_2 \gets \frac{36\sqrt{n} \cdot \log n}{\eps^2}$,
    $\theta \gets \frac{96\log n}{\eps^2}$  \label{ln:theta}\; 
    \smallskip
    
    $W_H \gets 0$, $W_L \gets 0$, $n_2 \gets 0$\;
    
    $S_1$: an array of size $t_1$, with entries initialized to $0$ \gray{\tcc{$S_1$ stores a set of witnesses}} 
    
    $S_2$ and $D$: arrays of size $t_2$, with entries initialized to $0$ \gray{\tcc{$S_2$ stores samples drawn from low-degree nodes, and $D$ stores their degrees}}
    
    \gray{\tcp{Pass $1$: obtain a ``witness'' set}}
    
    \ForEach{incoming $\sigma_j$} {
        \For{$i \gets 1, \dots, t_1$} {
            with probability $\frac{1}{j}$, $S_1[i] \gets \sigma_j$ 
        }
    }
    
    \gray{\tcp{Pass $2$: use the witness set to classify nodes as either high-degree or low-degree nodes, and estimate the contribution of high-degree nodes.}}

    \ForEach{incoming $\sigma_j$} {
        $Y_j \gets \sum_{i \in [t_1]} f(S_1[i], \sigma_j)$ \label{ln:sj}\; 
        \If{$Y_j \geq \theta$}{
            $W_H \gets W_H + \left(\frac{n}{t_1}\cdot Y_j\right)^{-1}$ \gray{\tcc{$\sigma_j$ is classified as a high-degree node; add the estimated contribution of $\sigma_j$ to $W_H$}}
        }
        \Else{
            \gray{\tcc{$\sigma_j$ is classified as a low-degree node}}
            $n_2 \gets n_2 + 1$ \;
            \lFor{$i \gets 1, \dots, t_2$} {
                with probability $\frac{1}{n_2}$, $S_2[i] \gets \sigma_j$     \gray{\tcc{obtain a sample of low-degree nodes}}
            }
        }
    }
    
    \gray{\tcp{Pass $3$: estimate the contribution of low-degree nodes by computing exact degrees of sampled low-degree nodes}}
    
    \ForEach{incoming $\sigma_j$} {
        \For{$i \gets 1, \dots, t_2$} {
            $D[i] \gets D[i] + f(S_2[i], \sigma_j)$
        }
    }   
    \gray{\tcc{rescale to estimate the total contribution of low-degree nodes}}
    $W_L \gets \frac{n_2}{t_2} \sum_{i \in [t_2]} \frac{1}{D[i]}$ 
    
    \Return $W_H + W_L$.
\end{algorithm}

We start with the diversity index DI. Our algorithm is described in Algorithm~\ref{alg:F0} with inline remarks.   

\vspace{2mm}
\noindent{\bf The Ideas and Algorithm Descriptions.\ }
The first idea one may come up with to estimate $\rFo(\sigma) = \sum_{j \in [n]} d_j^{-1}$ is to sample the summands $d_j^{-1}$.  However, the value of $d_j^{-1}$ spans the interval $[\frac{1}{n}, 1]$, meaning that $\Omega(n)$ samples are needed to estimate $\rFo$ up to an $O(1)$ factor. 

A natural idea to reduce the variance of the sampling approach is to partition the set of nodes into two groups, one, denoted by $H$, containing high-degree nodes (roughly, those with degree at least $\sqrt{n}$), and the other, denoted by $L$, containing low-degree nodes.  We then estimate their contributions to $\rFo$ separately.  However, in the streaming model, we cannot identify and store the two groups $H$ and $L$ explicitly using $o(n)$ space.  We resolve this issue by first sampling a {\em witness} set $S_1$ of size $\tilde{O}(\sqrt{n})$, which we store in memory. In the second pass, we use $S_1$ to classify all nodes as either high-degree nodes or low-degree nodes. For every high-degree node $\sigma_j$, we immediately calculate its estimated degree $\tilde{d}_j$, and add ${\tilde{d}_j}^{-1}$ to $W_H$; note that we do not store $\sigma_j$ in memory. We can show that $W_H = \sum_{j \in H} \tilde{d}_j^{-1}$ is a good estimate of $\sum_{j \in H} d_j^{-1}$ (i.e. the total contribution of $H$).  For low-degree nodes, we sample a subset $S_2 \subseteq L$ of size $\tilde{O}(\sqrt{n})$ and store it in memory.  In the third pass, we compute the degree $d_j$ of each node $\sigma_j$ in $S_L$ exactly, and then use $W_L = \frac{\abs{L}}{t_2} \sum_{\sigma_j \in S_2} d_j^{-1}$ to estimate $\sum_{j \in L} d_j^{-1}$.  We can show that every node in $L$ has a degree $O(\sqrt{n})$ with high probability. Therefore, a sample of size $\tilde{O}(\sqrt{n})$ is sufficient to estimate the total contribution of $L$ accurately.  Finally, we return $W_H + W_L$ as the estimation of $\rFo$. 

\vspace{2mm}
\noindent{\bf The Analysis.\ }
We now analyze Algorithm~\ref{alg:F0} and establish the following theorem.

\begin{theorem}
    \label{thm:F0}
    Given an input data stream $\sigma$ of length $n$ and a similarity function $f : U \times U \to [0, 1]$, for any $\eps \in (0, 1)$, Algorithm~\ref{alg:F0} computes a $(1+\eps, o(1))$-approximation of $\rFo(\sigma)$, using three passes and $O\left(\frac{1}{\eps^2}\sqrt{n}  \log n\right)$ words of space.
\end{theorem}

The space and pass complexities are clear from Algorithm~\ref{alg:F0}. In the remainder of this section, we show the correctness.

Recall that in Algorithm~\ref{alg:F0}, $Y_j$ is the total similarity weight between $S_1$ and $\sigma_j$. Let $H = \{ j \in [n] \mid Y_j \geq \theta \}$ and $L = \{ j \in [n] \mid Y_j < \theta \}$, where $\theta$ is defined at Line~\ref{ln:theta} of Algorithm~\ref{alg:F0}.  We have the following lemma.
\begin{lemma} 
    \label{lem:RF0-wH}
    With probability at least $1 - \frac{3}{n^3}$, the following events hold: (1) $\min_{j \in H} d_j \geq \sqrt{n}$, (2) $\max_{j \in L} d_j \leq 4 \sqrt{n}$, and (3) $W_H$ is a $(1+\eps)$-approximation of $\sum_{j \in H} d_j^{-1}$.
\end{lemma}

\begin{proof}
    
    Since all nodes in $S_1$ are sampled uniformly at random from $[n]$, for any $j \in [n]$ with $d_j < \sqrt{n}$, we have $\E[Y_j] = t_1 \cdot \frac{d_j}{n} < \frac{\theta}{2}$. By a Chernoff bound, 
    \begin{equation*}
        \Pr\left[ Y_j > \left(1 + \frac{\eps}{2} \right) \cdot \frac{\theta}{2} \right] \leq  e^{-\frac{(\theta/2)\cdot (\eps/2)^2}{3}} = \frac{1}{n^4}.
    \end{equation*}
    Therefore, with probability $1 - n \cdot \frac{1}{n^4} = 1 - \frac{1}{n^3}$, the following event $\cE_L$ holds: for all $j \in [n]$ with $d_j < \sqrt{n}$, $Y_j \le (1 + \frac{\eps}{2}) \cdot \frac{\theta}{2} < \theta$; consequently, $j \in L$. This implies that if $j \in H$, then $d_j \ge \sqrt{n}$, proving the first item.
    
    On the other hand, for any $j \in [n]$ with $d_j \ge \sqrt{n}$, we have $\E[Y_j] = t_1 \cdot \frac{d_j}{n} \ge \frac{\theta}{2}$.  By a Chernoff bound, 
    \begin{equation*}
        \Pr\left[ \abs{ Y_j - \E[Y_j] } > \frac{\eps}{2} \cdot \E[Y_j] \right] \leq 2 \exp\left(-\frac{\E[Y_j]\cdot (\eps / 2)^2}{3}\right) \leq 2 \exp\left(-\frac{(\theta/2)\cdot (\eps/2)^2}{3}\right) = \frac{2}{n^4}.
    \end{equation*}
    Therefore, with probability $1 - n \cdot \frac{2}{n^4} = 1 - \frac{2}{n^3}$, the following event $\cE_A$ holds: for all $j \in [n]$ with $d_j \ge \sqrt{n}$, $\frac{n Y_j}{t_1}$ is a $(1 + \frac{\eps}{2})$-approximation of $d_j$. Consequently, $\left(\frac{n Y_j}{t_1}\right)^{-1}$ is a $(1+ \eps)$-approximation of $d_j^{-1}$.
    
    We condition on events $\cE_L, \cE_A$ in the rest of the proof, which hold with probability $1 - \frac{1}{n^3} - \frac{2}{n^3} = 1 - \frac{3}{n^3}$.
    
    If $j \in L$, then either (1) $d_j < \sqrt{n}$, or (2) $d_j \geq \sqrt{n}$ and $Y_j < \theta$, where the second case implies $d_j \le \frac{n}{t_1} \cdot \frac{ \theta}{1-\frac{\eps}{2}}$, since $\frac{n Y_j}{t_1}$ is a $(1 + \frac{\eps}{2})$-approximation of $d_j$.  Therefore,
    \begin{equation*}
        \max_{j \in L} d_j \leq \max\left\{\sqrt{n}, \frac{n}{t_1} \cdot \frac{ \theta}{1-\frac{\eps}{2}}\right\} \leq 4 \sqrt{n},
    \end{equation*}
    proving the second item.
    
    If $j \in H$, then $d_j \geq \sqrt{n}$, and thus $\left(\frac{n Y_j}{t_1}\right)^{-1}$ is a $(1+ \eps)$-approximation of $d_j^{-1}$. As a result, $W_H=\sum_{j \in H} \left(\frac{n Y_j}{t_1}\right)^{-1}$ is a $(1+ \eps)$-approximation of $\sum_{j \in H} d_j^{-1}$, proving the third item.
\end{proof}

Let $\cE_1$ denote the event that Lemma~\ref{lem:RF0-wH} holds, then $\Pr[\cE_1] \geq 1 - \frac{3}{n^3}$. The next lemma shows that, conditioned on $\cE_1$, $W_L$ provides a good approximation to $\sum_{j \in L} d_j^{-1}$.

\begin{lemma} \label{lem:RF0-wL}
    Conditioned on $\cE_1$, with probability at least $1 - \frac{2}{n^3}$, $W_L$ is a $(1 + \eps)$-approximation of $\sum_{j \in L} d_j^{-1}$.
\end{lemma}

\begin{proof}
    For each $i \in [t_2]$, let $X_i$ be the $i$-th sample in $S_2$ and $Z_i = d_{X_i}^{-1}$. Fix any $S_1$ such that $\cE_1$ holds. Since each $X_i$ is sampled uniformly at random from $L$, $\E[Z_i \mid S_1] = \frac{1}{\abs{L}} \sum_{j \in L} d_j^{-1}$. By Lemma~\ref{lem:RF0-wH}, for all $i \in [t_2]$ we have 
    \begin{equation*}
        \frac{1}{4\sqrt{n}} \leq \E[Z_i \mid S_1] \leq 1.
    \end{equation*}
    Let $Z = \frac{1}{t_2} \sum_{i \in [t_2]} Z_i$. We have $\E[Z \mid S_1] = \E[Z_i \mid S_1] \geq \frac{1}{4 \sqrt{n}}$. By a Chernoff bound,
    \begin{equation*}
        \Pr\left[ \abs{Z - \E[Z \mid S_1]} > \eps \E[Z \mid S_1] \mid S_1 \right] \leq 2 \exp\left(- \frac{\eps^2 t_2  \E[Z \mid S_1]}{3} \right) \leq \frac{2}{n^3}.
    \end{equation*}
    Therefore, with probability $\left(1 - \frac{2}{n^3}\right)$, $W_L = \abs{L} \cdot Z$ is a $(1 + \eps)$-approximation of the quantity $\abs{L} \cdot \E[Z \mid S_1] = \sum_{j \in L} d_j^{-1}$.
\end{proof}

By Lemma~\ref{lem:RF0-wH} and Lemma~\ref{lem:RF0-wL}, with probability $\Pr[\cE_1] \cdot \left(1 - \frac{2}{n^3}\right) = 1 - o(1)$, $W_H$ and $W_L$ are $(1+\eps)$-approximations of $\sum_{j \in H} d_j^{-1}$ and $\sum_{j \in L} d_j^{-1}$, respectively. Since $\rFo = \sum_{j \in H} d_j^{-1} + \sum_{j \in L} d_j^{-1}$,  $W_H + W_L$ is a $(1+\eps)$-approximation of $\rFo$.

\subsection{$\rLo$-Sampling}
\label{sec:ell0}

In this section, we consider $\rLo$-sampling. Our algorithm is described in Algorithm~\ref{alg:ell0}.

\begin{algorithm}[!t]
    \caption{$\rLo$-sampling}
    \label{alg:ell0}
    \DontPrintSemicolon
    \SetAlgoNoEnd
    
    \KwIn{a data stream $\sigma$ of size $n$ and a similarity function $f : U \times U \to [0, 1]$, parameter $\eps$}
    \KwOut{a node in $\sigma$ sampled with probability inversely proportional to its degree}
    
    \smallskip
    $t \gets 12\sqrt{n}\log n$,     $\theta \gets \frac{96\log n}{(\eps/4)^2}$\;
    
    $n_H \gets 0$, $n_L \gets 0$, $r_H \gets \Null, r_L \gets \Null$\;
    
    $S_H, D_H, S_L, D_L$: arrays of size $t$, with entries initialized to 0 \gray{\tcc{$S_H$ and $S_L$ are samples drawn from high-degree and low-degree nodes, respectively, with their corresponding degrees stored in $D_H$ and $D_L$}}
    
    \gray{\tcp{Pass $1$: identify the witness set}}
    
    run Pass 1 of Algorithm~\ref{alg:F0} with input $(\sigma, \frac{\eps}{4})$ to obtain the witness set $S_1$
    
    \gray{\tcp{Pass $2$: classify nodes as either high-degree or low-degree nodes and obtain a sample from each group, denoted by $S_H$ and $S_L$}}
    
    parallel run Pass $2$ of Algorithm~\ref{alg:F0}, computing for each node $\sigma_j$ the total similarity weight (denoted by $Y_j$) in $S_1$, see Line~\ref{ln:sj} in Algorithm~\ref{alg:F0}, and get $W_H$, the contribution of high-degree nodes to $\rFo$
    
    \ForEach{incoming $\sigma_j$} {
        \uIf{$Y_j \geq \theta$} {
            $n_H \gets n_H + 1$\;
            \lFor{$i \gets 1, \dots, t$} {
                with probability $\frac{1}{n_H}$, $S_H[i] \gets \sigma_j$ 
            }
        }
        \Else {
            $n_L \gets n_L + 1$\;
            \lFor{$i \gets 1, \dots, t$} {
                with probability $\frac{1}{n_L}$, $S_L[i] \gets \sigma_j$ 
            }
        }
    }
    
    \gray{\tcp{Pass $3$: get a sample using $S_H$ and $S_L$}}
    
    parallel run Pass $3$ of Algorithm~\ref{alg:F0} and get $W_L$, the contribution of low-degree nodes to $\rFo$ 
    
    \gray{\tcc{compute the degrees of nodes in $S_H$ and $S_L$}} 
    \ForEach{incoming $\sigma_j$} {
        \For{$i \gets 1, \dots, t$} {
            $D_H[i] \gets D_H[i] + f(S_H[i], \sigma_j)$\;
            $D_L[i] \gets D_L[i] + f(S_L[i], \sigma_j)$\;
        }
    }
    
    \gray{\tcc{get one sample from each of $S_H$ and $S_L$ and then ``merge''}} \label{ln:getRHRL}
    \For{$i \gets 1, \ldots, t$} {
        with probability $\min\left\{1, \frac{\sqrt{n}}{D_H[i]}\right\}$, $r_H \gets S_H[i]$, \Break if $r_H$ is assigned
    }
    \For{$i \gets 1, \ldots, t$} {
        with probability $\frac{1}{D_L[i]}$, $r_L \gets S_L[i]$, \Break if $r_L$ is assigned
    }
    \label{ln:getRHRLend}
    
    \If{$(r_H = \Null \ \text{and}\ n_H > 0)$ or $(r_L = \Null \ \text{and}\ n_L > 0)$} {
        \Return Fail
    }
    \Else{
        \Return $r_H$ with probability $\frac{W_H}{W_H + W_L}$ or $r_L$ otherwise.
    }
\end{algorithm}

\vspace{2mm}
\noindent{\bf The Ideas and Algorithm Descriptions.\ }
We begin with a special case: when $f$ is the equality function, $\DI$ coincides with $F_0$, the number of distinct elements in $\sigma$, and $\rLo$-sampling coincides with $\ell_0$-sampling. A simple approach for $\ell_0$-sampling is to sample each {\em distinct} element with probability $p \approx \frac{1}{F_0}$, which ensures a good probability of sampling exactly one element. Given this outcome, the sampled element is uniformly distributed over all distinct elements in the data stream. For a general similarity function, a natural adaptation of this approach is to sample each node $\sigma_j$ with probability $p_j \approx \frac{1}{d_j} \cdot \frac{1}{\rFo}$, which ensures a good probability that the sample set $S$ contains exactly one node.  Unfortunately, a node sampled in this manner does not necessarily follow the desired distribution.

For example, consider a graph consisting of three nodes $\{1, 2, 3\}$ and five edges (including self-loops) $\{(1, 1), (1, 2), (1, 3), (2, 2), (3, 3)\}$. In this case, $\rFo = \frac{1}{3} + \frac{1}{2} + \frac{1}{2} = \frac{4}{3}$ and the probability of choosing Node $1$ according to the $\rLo$-distribution is $\frac{1/3}{\rFo} = \frac{1}{4}$.  On the other hand, if we first sample the three nodes independently into $S$ with probability $\frac{1/3}{\rFo} = \frac{1}{4}$, $\frac{1/2}{\rFo} = \frac{3}{8}$ and $\frac{1/2}{\rFo} = \frac{3}{8}$, and output the node if exactly one node is sampled (i.e., $\abs{S} = 1$), the probability that Node $1$ is sampled is
\begin{equation*}
    \Pr(S=\{1\} \mid \abs{S}=1) = \frac{\Pr(S=\{1\})}{\Pr(\abs{S}=1)} = \frac{1/4 \cdot 5/8 \cdot 5/8}{1/4 \cdot 5/8 \cdot 5/8 + 3/4 \cdot 3/8 \cdot 5/8 \cdot 2} = \frac{5}{23},
\end{equation*}
which is not equal to $\frac{1}{4}$. 

We thus take a different approach: we first sample a node $\sigma_j$ uniformly at random, and then accept it with probability ${1}/{d_j}$.  To ensure successful acceptance, we run multiple Reservoir sampling structures in parallel, and at the end of the streaming process, repeatedly apply the rejection sampling procedure until a sample is accepted.  The problem with this vanilla rejection sampling method is that it may require $\Omega(n)$ samples to obtain a single accepted sample, resulting in $\Omega(n)$ space complexity.  

To reduce the space usage, we again partition the $n$ nodes into high-degree group $H$ and low-degree group $L$ as in Algorithm~\ref{alg:F0}. We next perform rejection sampling on two groups with {\em different} rejection rates, and then properly merge the accepted samples in the two groups at the end.   More precisely, in the first pass, we again try to sample a witness set, which is used in the second pass to partition the $n$ nodes into two groups using a degree-threshold roughly $\sqrt{n}$.  At the same time, we obtain a random subset of size $t = \tilde{O}(\sqrt{n})$ from each group. In the third pass, we compute the degree of the sampled nodes, and then perform rejection samplings on the two sample sets.  For each node $\sigma_j$ in the sample of the high-degree group, we use a rejection rate of $\min\left\{1, \frac{\sqrt{n}}{d_j}\right\}$, where we put a cap $1$ in the minimization since the value $\frac{\sqrt{n}}{d_j}$ may be larger than $1$ due to the small inaccuracy of the high-low classification.  For each node $\sigma_j$ in the sample of the low-degree group, we use a rejection rate of $\frac{1}{d_j}$.  When we get one accepted sample from each group, we select one of the two accepted samples with a probability proportional to the $\rFo$ contributions of the high-degree and low-degree groups, which can be approximated by running Algorithm~\ref{alg:F0} in parallel.

\vspace{2mm}
\noindent{\bf The Analysis.\ }
We now analyze Algorithm~\ref{alg:ell0} and prove the following theorem.

\begin{theorem} \label{thm:L0}
    Given an input data stream of length $n$ and a similarity function $f : U \times U \to [0, 1]$, for any $\eps \in (0, 1)$, Algorithm~\ref{alg:ell0} is a $(1+\eps, o(1))$-$\rLo$-sampler; it uses three passes and $O\left(\frac{1}{\eps^2}\sqrt{n} \log n\right)$ words of space.
\end{theorem}

The pass and space complexities are clear from Algorithm~\ref{alg:ell0}.  In the remainder of this section, we show the correctness.

W.l.o.g., assume $\eps \ge n^{-1/4}$; otherwise $\frac{\sqrt{n}\log n}{\eps^2} > n$, in which case we can store all nodes of the graph and perform the sampling offline.

Let $\cE_2$ be the event that both Lemma~\ref{lem:RF0-wH} and Lemma~\ref{lem:RF0-wL} hold. Then $\Pr[\cE_2] \geq 1 - O(n^{-3})$. For each $i \in [t]$, let $X_i$ be the $i$-th sample in $S_H$ and $X_i^{\prime}$ be the $i$-th sample in $S_L$. We have the following lemma.

\begin{lemma} \label{lem:RL0-fail}
    Conditioned on $\cE_2$, Algorithm~\ref{alg:ell0} fails with probability at most ${2}/{n^3}$.
\end{lemma}

\begin{proof}
    By Lemma~\ref{lem:RF0-wH}, conditioned on $\cE_2$, for each $i \in [t]$, we have
    $\sqrt{n} \leq d_{X_i} \leq n$ and $1 \leq d_{X_i^{\prime}} \leq 4\sqrt{n}$.
    If $H \neq \emptyset$, then
    \begin{equation}
        \Pr[r_H = \Null \mid \cE_2] = \prod_{i \in [t]} \left(1 - \min\left\{1, \frac{\sqrt{n}}{d_{X_i}} \right\} \right) \leq \left(1 - \frac{1}{\sqrt{n}}\right)^t \leq e^{-\frac{t}{\sqrt{n}}} = \frac{1}{n^{12}}. \label{eq:RL0-fail1}
    \end{equation}
    Similarly, if $L  \neq \emptyset$, then
    \begin{equation}
        \Pr[r_L = \Null \mid \cE_2] = \prod_{i \in [t]} \left(1 - \frac{1}{d_{X'_i}}\right) \leq \left(1 - \frac{1}{4\sqrt{n}}\right)^t \leq e^{-\frac{t}{4\sqrt{n}}} = \frac{1}{n^3}. \label{eq:RL0-fail2}
    \end{equation}
    By \eqref{eq:RL0-fail1} and \eqref{eq:RL0-fail2}, with probability at most ${1}/{n^{12}} + {1}/{n^3} < {2}/{n^3}$, either $(r_H = \Null \ \land\ H \neq \emptyset)$ or $(r_L = \Null \ \land\ L \neq \emptyset)$ occurs, causing Algorithm~\ref{alg:ell0} to fail.
\end{proof}

\begin{lemma} \label{lem:RL0-success}
    Conditioned on $\cE_2$, Algorithm~\ref{alg:ell0} samples each node $j$ with probability in the range of $(1 \pm \eps) {d_j^{-1}}/{\rFo}$.
\end{lemma}

\begin{proof}
    
    We first show that conditioned on not being $\Null$, $r_H$ follows the distribution $\left\{\frac{d_j^{-1}}{\sum_{j^{\prime} \in H} (d_{j'})^{-1}}\right\}_{j \in H}$.
    
    For any $j \in H$ and $i \in [t]$, since $X_i$ is sampled from $H$ uniformly at random, we have
    \begin{eqnarray*}
        \Pr\left[r_H = j, r_H = X_i \mid \cE_2\right] &=& \Pr\left[X_i = j \mid \cE_2 \right] \cdot \Pr\left[r_H = X_i \mid X_i = j, \cE_2\right] \\
        &=& \frac{1}{\abs{H}} \cdot \frac{\sqrt{n}}{d_j} \cdot \Pr[\forall k \in [i-1], r_H \neq X_k  \mid \cE_2].
    \end{eqnarray*}
    It follows that
    \begin{eqnarray*}
        \Pr\left[r_H = j \mid r_H \neq \Null, \cE_2\right] &=& \frac{\Pr\left[r_H = j \mid \cE_2\right]}{\Pr[r_H \neq \Null \mid \cE_2]} \\
        &=& \frac{\sum_{i \in [t]} \Pr\left[r_H = j, r_H = X_i \mid \cE_2\right]}{\Pr[r_H \neq \Null \mid \cE_2]} \\ 
        &=& Z \cdot d_j^{-1},
    \end{eqnarray*}
    where $Z = \frac{\sqrt{n} \sum_{i \in [t]}  \Pr[\forall k \in [i-1], r_H \neq X_k \mid \cE_2]}{\abs{H} \cdot \Pr[r_H \neq \Null \mid \cE_2]} $ does not depend on $j$. Since the probabilities sum to 1, we have
    \begin{equation}
        \Pr\left[r_H = j \mid r_H \neq \Null, \cE_2\right] = \frac{d_j^{-1}}{\sum_{j^{\prime} \in H} (d_{j'})^{-1}}. \label{eq:RL0-1}
    \end{equation}
    By a similar argument, for any $j \in L$,
    \begin{equation}
        \Pr\left[r_L = j \mid r_L \neq \Null, \cE_2\right] = \frac{d_j^{-1}}{\sum_{j^{\prime} \in L} (d_{j'})^{-1}}. \label{eq:RL0-2}
    \end{equation}
    
    Without loss of generality, we assume that both $H$ and $L$ are not empty, otherwise Algorithm~\ref{alg:ell0} is already a $\rLo$-sampler by \eqref{eq:RL0-1} and \eqref{eq:RL0-2}. By Lemma~\ref{lem:RF0-wH}, Lemma~\ref{lem:RF0-wL}, and our choices of parameters, $W_H$ and $W_L$ are $\left(1 + \frac{\eps}{4}\right)$-approximation of $\sum_{j \in H} d_j^{-1}$ and $\sum_{j \in L} d_j^{-1}$, respectively. Therefore,
    \begin{equation}
        \frac{W_H}{W_H + W_L} \in \left[\frac{1-(\eps /4)}{1 + (\eps /4)}, \frac{1+(\eps /4)}{1 - (\eps /4)} \right] \cdot \frac{\sum_{j \in H} d_j^{-1}}{\rFo} = (1 \pm \eps) \cdot \frac{\sum_{j \in H} d_j^{-1}}{\rFo}. \label{eq:RL0-3}
    \end{equation}
    By \eqref{eq:RL0-1} and \eqref{eq:RL0-3}, for any $j \in H$, we have
    \begin{eqnarray*}
        \Pr\left[\text{output } j \mid \text{no failure}, \cE_2\right] = \frac{W_H}{W_H + W_L} \cdot \frac{d_j^{-1}}{\sum_{j^{\prime} \in H} (d_{j'})^{-1}} = (1 \pm \eps) \frac{d_j^{-1}}{\rFo}.
    \end{eqnarray*}
    By a similar argument, for any $j \in L$, 
    \begin{eqnarray*}
        \Pr\left[\text{output } j \mid \text{no failure}, \cE_2\right] = \frac{W_L}{W_H + W_L} \cdot \frac{d_j^{-1}}{\sum_{j^{\prime} \in L} (d_{j'})^{-1}} = (1 \pm \eps) \frac{d_j^{-1}}{\rFo}.
    \end{eqnarray*}
\end{proof}

By Lemma~\ref{lem:RF0-wH}, Lemma~\ref{lem:RF0-wL} and Lemma~\ref{lem:RL0-fail}, the total failure probability of Algorithm~\ref{alg:ell0} is at most $\Pr[\cE_2, \text{failure}] + \Pr[\lnot \cE_2] = O(n^{-3})$. By Lemma~\ref{lem:RL0-fail} and Lemma~\ref{lem:RL0-success}, for any $j \in [n]$, we have
\begin{eqnarray*}
    \Pr\left[\text{output } j, \text{no failure}\right] &=& \Pr\left[\text{output } j, \text{no failure}, \cE_2\right] + \Pr\left[\text{output } j, \text{no failure}, \lnot \cE_2\right] \\
    &=& \Pr[\cE_2, \text{no failure}] \cdot \Pr\left[\text{output } j \mid \text{no failure}, \cE_2\right] \pm \Pr[\lnot \cE_2] \\
    &=& (1 \pm \eps) \cdot \frac{d_j^{-1}}{\rFo} \pm O(n^{-3}),
\end{eqnarray*}
which implies that
\begin{eqnarray*}
    \Pr\left[\text{output } j \mid \text{no failure}\right] &=& \frac{\Pr\left[\text{output } j, \text{no failure}\right]}{\Pr\left[\text{no failure}\right]} \\
    &=& (1 \pm \eps) \cdot \frac{d_j^{-1}}{\rFo} \pm O(n^{-3}) = (1 \pm 2\eps) \cdot \frac{d_j^{-1}}{\rFo},
\end{eqnarray*}
where the last equality holds since we have assumed that $\eps \ge n^{-1/4}$, and $d_j \le n$ and $\rFo \le n$.

Consequently, by down-scaling $\eps$ by a factor of $2$, Algorithm~\ref{alg:ell0} is a $(1+\eps, o(1))$-$\rLo$-sampler using $O\left(\frac{1}{\eps^2}\sqrt{n} \log n\right)$ words of space.

\subsection{$\rFp$-Estimation}
\label{sec:Fp}

In this section, we study $\rFp$-estimation. Our algorithm is described in Algorithm~\ref{alg:Fp}.  

\begin{algorithm}[!t]
    \caption{$\rFp$-estimation}
    \label{alg:Fp}
    \DontPrintSemicolon
    \SetAlgoNoEnd
    
    \KwIn{a data stream $\sigma$ of size $n$ and a similarity function $f : U \times U \to [0, 1]$, parameter $\eps$}
    \KwOut{a $(1 +\eps)$-approximation of $\rFp(\sigma)$}
    
    \smallskip
    $\eps^{\prime} \gets \frac{\eps}{4(p+1)}$, $t_1 \gets \frac{6n^{1-1/(p+1)} \cdot \log n}{(\eps^{\prime})^2}$, $t_2 \gets \frac{4^{p+2}n^{1-1/(p+1)}}{\eps^2}$, $\theta \gets \frac{12\log n}{(\eps^{\prime})^2}$  \label{ln:theta-fp}\;
     
    $W_H \gets 0$, $W_L \gets 0$ \;
    
    $S_1$: an array of size $t_1$,  with entries initialized to $0$  \gray{\tcc{$S_1$ stores a set of witnesses}} 
    
    $S_2$, $I_L$ and $D$: arrays of size $t_2$, with entries initialized to $0$  \gray{\tcc{$S_2$ contains samples drawn from $[n]$, $I_L$ records whether each node in $S_2$ belongs to $L$, and $D$ stores their degrees}}
    
    \gray{\tcp{Pass $1$: obtain samples $S_1$ (witness set) and $S_2$ (low-degree candidates)}}
    
    \ForEach{incoming $\sigma_j$} {
        \lFor{$i \gets 1, \dots, t_1$} {
            with probability $\frac{1}{j}$, $S_1[i] \gets \sigma_j$ 
        }
        \lFor{$i \gets 1, \dots, t_2$} {
            with probability $\frac{1}{j}$, $S_2[i] \gets \sigma_j$ 
        }
    }
    
    \gray{\tcp{Pass $2$:  estimate $\rFp$}}
    
    \ForEach{incoming $\sigma_j$} {
        \For{$i \gets 1, \dots, t_2$} {
            $D[i] \gets D[i] + f(S_2[i], \sigma_j)$ \gray{\tcc{compute the degrees of nodes in $S_2$}}
        }
        $Y_j \gets \sum_{i \in [t_1]} f(S_1[i], \sigma_j)$\;
        \If{$Y_j \geq \theta$} {
            $W_H \gets W_H + \left(\frac{n}{t_1}\cdot Y_j\right)^p$ \gray{\tcc{$\sigma_j$ is classified as a high-degree node; add the estimated contribution of $\sigma_j$ to $W_H$}}
        }
        \Else{
            \gray{\tcc{$\sigma_j$ is classified as a low-degree node; mark all $\sigma_j \in S_2$}}
            
            \For{$i \gets 1, \dots, t_2$} {
                \lIf{$S_2[i] = \sigma_j$} {
                    $I_L[i] \gets 1$ 
                }
            }
        }
    }
    \gray{\tcc{rescale to estimate the total contribution of low-degree nodes}}
    
    $W_L \gets \frac{n}{t_2} \sum_{i \in [t_2]} \left(I_L[i] \cdot D[i]^p\right)$ 
    
    \Return $W_H + W_L$.
\end{algorithm}

\vspace{2mm}
\noindent{\bf The Ideas and Algorithm Descriptions.\ }
We can use the same strategy as Algorithm~\ref{alg:F0} to estimate $\rFp$.  The idea is to partition the nodes into a high-degree group and a low-degree group based on a degree threshold of approximately $n^{1/(p+1)}$, using a witness set $S_1$ sampled in the first pass.  Like before, we compute the contribution of the high-degree group by approximating each node's degree in the second pass, while the low-degree group's contribution is derived from a sample set with exact degrees in the third pass.  As before, this approach needs three passes.  

However, for $\rFp$, we can compress three passes to two passes.  The key observation is that we can prepare for a sample $S_2$ of the low-degree group in the first pass {\em without} identifying the low-degree nodes.  We can then filter out high-degree nodes in $S_2$ in the second pass and use the remaining nodes to estimate the contribution of low-degree nodes. 
If, after filtering, the number of remaining nodes in $S_2$ is very small, then we know that the number of low-degree nodes in the whole dataset is small, and thus the contribution of low-degree nodes is negligible. Otherwise, we will obtain sufficiently many samples to accurately estimate the contribution of the low-degree nodes. This idea does {\em not} work for $\rFo$, since the contribution of a small number of low-degree nodes to $\rFo$ may be dominant.

\vspace{2mm}
\noindent{\bf The Analysis.\ } 
We now analyze Algorithm~\ref{alg:Fp} and prove the following theorem.

\begin{theorem} \label{thm:Fp}
    Given an input data stream of length $n$ and a similarity function $f : U \times U \to [0, 1]$, for any $\eps \in (0, 1)$,   there is an algorithm that computes a $(1+\eps, 0.01)$-approximation of $\rFp\ (p > 0)$, using two passes and $O\left(\frac{1}{\eps^2} n^{1-1/(p+1)} \log n\right)$ words of space.
\end{theorem}

The pass and space complexities are clear from Algorithm~\ref{alg:Fp}. We now prove that with probability at least $2/3$, Algorithm~\ref{alg:Fp} outputs a $(1+\eps)$-approximation of $\rFp$. By parallel repetition and taking the median, we can boost the success probability to 0.99 with an additional constant factor in the space bound.

Recall that in Algorithm~\ref{alg:Fp}, $Y_j$ is the total similarity weight between $S_1$ and $\sigma_j$. Let $H = \{ j \in [n] \mid Y_j \geq \theta \}$ and $L = \{ j \in [n] \mid Y_j < \theta \}$, where $\theta$ is defined at Line~\ref{ln:theta-fp} of Algorithm~\ref{alg:Fp}.  Let $W_H$ and $W_L$ be the estimated contributions of $H$ and $L$, respectively.  

We have the following lemma.
\begin{lemma} 
    \label{lem:RFp-wH}
    With probability $1 - o(1)$, the following events hold: (1) $\min_{j \in H} d_j \geq n^{1/(p+1)}$, (2) $\max_{j \in L} d_j \leq 4 n^{1/(p+1)}$, and (3) $W_H$ is a $\left(1+\frac{\eps}{2}\right)$-approximation of $\sum_{j \in H} d_j^p$.
\end{lemma}

\begin{proof}
    Since all nodes in $S_1$ are sampled uniformly at random from $[n]$, for any $j \in [n]$ with $d_j < n^{1/(p+1)}$, we have $\E[Y_j] = t_1 \cdot \frac{d_j}{n} < \frac{\theta}{2}$.  By a Chernoff bound, 
    \begin{equation*}
        \Pr\left[ Y_j > \left(1 + \eps^{\prime} \right) \cdot \frac{\theta}{2} \right] \leq  e^{-\frac{(\theta/2)\cdot (\eps^{\prime})^2}{3}} = \frac{1}{n^2}.
    \end{equation*}
    Therefore, with probability $1 - n \cdot \frac{1}{n^2} = 1 - \frac{1}{n}$, the following event $\cE'_L$ holds: for all $j \in [n]$ with $d_j < n^{1/(p+1)}$, $Y_j \le (1 + \eps^{\prime}) \cdot \frac{\theta}{2} < \theta$, and thus $j \in L$. This implies that if $j \in H$, then $d_j \ge n^{1/(p+1)}$, proving the first item.
    
    On the other hand, for any $j \in [n]$ with $d_j \ge n^{1/(p+1)}$, we have $\E[Y_j] = t_1 \cdot \frac{d_j}{n} \ge \frac{\theta}{2}$.  By a Chernoff bound, 
    \begin{equation*}
        \Pr\left[ \abs{ Y_j - \E[Y_j] } > \eps^{\prime} \cdot \E[Y_j] \right] \leq 2 e^{-\frac{\E[Y_j]\cdot (\eps^{\prime})^2}{3}} \leq 2 e^{-\frac{(\theta/2)\cdot (\eps^{\prime})^2}{3}} = \frac{2}{n^2}.
    \end{equation*}
    Therefore, with probability $1 - n \cdot \frac{2}{n^2} = 1 - \frac{2}{n}$, the following event $\cE'_A$ holds: for all $j \in [n]$ with $d_j \ge n^{1/(p+1)}$, $\frac{n Y_j}{t_1}$ is a $(1 + \eps^{\prime})$-approximation of $d_j$. Consequently, $\left(\frac{n Y_j}{t_1}\right)^p$ is a $\left(1+ \frac{\eps}{2}\right)$-approximation of $d_j^p$ when $p > 0$.
    
    We condition on events $\cE'_L$ and $\cE'_A$ in the rest of the proof, which hold with probability $1 - o(1)$.
    
    If $j \in L$, then either (1) $d_j < n^{1/(p+1)}$, or (2) $d_j \geq n^{1/(p+1)}$ and $Y_j < \theta$, where the second case implies $d_j \le \frac{n}{t_1} \cdot \frac{ \theta}{1 - \eps^{\prime}}$, since $\frac{n Y_j}{t_1}$ is a $(1 + \eps^{\prime})$-approximation of $d_j$.  Therefore,
    \begin{equation*} 
        \max_{j \in L} d_j \leq \max\left\{n^{1/(p+1)}, \frac{n}{t_1} \cdot \frac{ \theta}{1-\eps^{\prime}}\right\} \leq 4 n^{1/(p+1)},
    \end{equation*}
    proving the second item.
    
    If $j \in H$, then $d_j \geq n^{1/(p+1)}$, and thus $\left(\frac{n Y_j}{t_1}\right)^p$ is a $\left(1+ \frac{\eps}{2}\right)$-approximation of $d_j^p$. As a result, $W_H=\sum_{j \in H} \left(\frac{n Y_j}{t_1}\right)^p$ is a $\left(1+ \frac{\eps}{2}\right)$-approximation of $\sum_{j \in H} d_j^p$, proving the third item.
\end{proof}

Let $\cE_3$ denote the event that Lemma~\ref{lem:RFp-wH} holds, then $\Pr[\cE_3] = 1 - o(1)$. The next lemma shows that, conditioned on $\cE_3$, $W_L$ provides an approximation to $\sum_{j \in L} d_j^p$ up to a small additive error.

\begin{lemma} \label{lem:RFp-wL}
    Conditioned on $\cE_3$, with probability at least $\frac{3}{4}$, we have
    \begin{equation*}
        \sum_{j \in L} d_j^p - \frac{\eps}{2}\cdot \rFp \leq W_L \leq \sum_{j \in L} d_j^p + \frac{\eps}{2}\cdot \rFp.
    \end{equation*}
\end{lemma}

\begin{proof}
    For each $i \in [t_2]$, let $X_i$ be the $i$-th sample in $S_2$ and $Z_i = \mathbf{1}\{X_i \in L\} \cdot n \cdot d_{X_i}^p$. Fix any $S_1$ such that $\cE_3$ holds. Since each $X_i$ is sampled from $[n]$ uniformly at random, we have $\E[Z_i \mid S_1] = \frac{1}{n} \sum_{j \in L} n \cdot d_j^p = \sum_{j \in L} d_j^p$. Furthermore, by Lemma~\ref{lem:RFp-wH},  for all $i \in [t_2]$ we have
    \begin{equation*}
        \Var[Z_i \mid S_1] \leq \E[Z_i^2 \mid S_1] \leq n\cdot (4n^{1/(p+1)})^p \cdot  \E[Z_i \mid S_1] = 4^p n^{2-1/(p+1)} \cdot \sum_{j \in L} d_j^p.
    \end{equation*}
    Let $Z = \frac{1}{t_2} \sum_{i \in [t_2]} Z_i$. Observe that \begin{equation*}
        \Var[Z \mid S_1] = \frac{\Var[Z_i \mid S_1]}{t_2} \leq \frac{\eps^2 n}{16} \cdot \sum_{j \in L} d_j^p.
    \end{equation*}
    By Chebyshev's inequality,
    \begin{equation*}
        \Pr\left[ \abs{Z - \E[Z \mid S_1]} > \frac{\eps}{2} \cdot \rFp \ \middle| \ S_1 \right] \leq \frac{4\Var[Z \mid S_1]}{\eps^2 \rFp^2} \leq \frac{n \cdot \sum_{j \in L} d_j^p}{4 \rFp^2} \leq \frac{1}{4}.
    \end{equation*}
    The last inequality holds by noting that $\sum_{j \in L} d_j^p \leq \rFp$ and $n \leq \rFp$. 
    
    Since $W_L = Z$ and $\E[Z \mid S_1] = \sum_{j \in L} d_j^p$, with probability at least $\frac{3}{4}$, we have
    \begin{equation*}
        \sum_{j \in L} d_j^p - \frac{\eps}{2}\cdot \rFp \leq W_L \leq \sum_{j \in L} d_j^p + \frac{\eps}{2}\cdot \rFp.
    \end{equation*}
\end{proof}

By Lemma~\ref{lem:RFp-wH} and Lemma~\ref{lem:RFp-wL}, with probability $(1-o(1))\cdot \frac{3}{4} \geq \frac{2}{3}$, we have
\begin{equation}
    \left(1 - \frac{\eps}{2}\right) \sum_{j \in H} d_j^p \leq W_H \leq \left(1 + \frac{\eps}{2}\right) \sum_{j \in H} d_j^p, \label{eq:Fphigh}
\end{equation}
and
\begin{equation}
    \sum_{j \in L} d_j^p - \frac{\eps}{2} \cdot \rFp \leq W_L \leq \sum_{j \in L} d_j^p + \frac{\eps}{2} \cdot \rFp. \label{eq:Fplow}
\end{equation}
Adding \eqref{eq:Fphigh} and \eqref{eq:Fplow} gives
\begin{equation*}
    \rFp - \frac{\eps}{2}\left(\rFp + \sum_{j \in H} d_j^p\right) \leq W_H + W_L \leq \rFp + \frac{\eps}{2}\left(\rFp + \sum_{j \in H} d_j^p \right).
\end{equation*}
Since $\sum_{j \in H} d_j^p \leq \rFp$, with probability at least $\frac{2}{3}$, Algorithm~\ref{alg:Fp} returns a $(1 + \eps)$-approximation of $\rFp$.

\subsection{$L_{M_p}$-Sampling}
\label{sec:LMp-sampling}

In this section, we consider $L_{M_p}$-sampling. Our algorithm is described in
Algorithm~\ref{alg:LMp-sampling}.

\begin{algorithm}[!t]
    \caption{$L_{M_p}$-Sampling}
    \label{alg:LMp-sampling}
    \DontPrintSemicolon
    \SetAlgoNoEnd
    
    \KwIn{a data stream $\sigma$ of size $n$ and a similarity function $f$, parameter $\eps$}
    \KwOut{a node in $\sigma$ sampled with probability proportional to the $p$-th power of its degree}
    
    $\eps'\gets \eps/(16(p+1))$,
    $t_1\gets \frac{6(p+3)n^{1-1/(p+1)}\log n}{(\eps')^2}$,
    $t_2\gets 12\cdot 4^p(p+3)n^{1-1/(p+1)}\log n$\;
    $\theta\gets \frac{12(p+3)\log n}{(\eps')^2}$,
    $\lambda\gets \frac{6(p+3)\log n}{n}$,
    $K\gets\infty$, $r\gets\Null$\;
    $S_1$: an array of size $t_1$, with entries initialized to $0$
    \gray{\tcc{$S_1$ stores a set of witnesses}}
    $S_2$, $I_L$, and $D$: arrays of size $t_2$, with entries initialized to $0$
    \gray{\tcc{$S_2$ stores low-degree candidates, $I_L$ records whether a candidate belongs to $L$ and $D$ stores its degree}}
    
    \gray{\tcp{Pass 1: obtain the witness set and low-degree candidates}}
    \ForEach{incoming $\sigma_j$}{
        \For{$i\gets 1, \dots, t_1$}{
            \text{with probability }$1/j$, $S_1[i]\gets\sigma_j$\;
        }
        draw $R_j\sim \EXP(1)$\;
        \If{$R_j\le (4n^{1/(p+1)})^p\lambda$}{
            insert $\sigma_j$ into an empty entry of $S_2$; if no such entry exists, \Return Fail \label{ln:LMp-fail}
        }
    }
    
    \gray{\tcp{Pass 2: find the minimum of $R_j / d_j^p$ over all high-degree nodes and stored low-degree nodes}}
    \ForEach{incoming $\sigma_j$}{
        regenerate $R_j$ using the same random seed\;
        \For{$i\gets 1, \dots, t_2$}{
            $D[i]\gets D[i]+f(S_2[i],\sigma_j)$ \gray{\tcc{compute the degrees of nodes in $S_2$}}
        }
        $Y_j\gets\sum_{i\in[t_1]}f(S_1[i],\sigma_j)$\;
        \eIf{$Y_j\ge\theta$}{
            $\widetilde d_j\gets nY_j/t_1$ \gray{\tcc{$\sigma_j \in H$; compute the estimated degree of $\sigma_j$ }}
            \If{$R_j/\widetilde d_j^p<K$}{
                $(K,r)\gets(R_j/\widetilde d_j^p, \sigma_j)$ \gray{\tcc{update the current minimum using $\sigma_j$ }}
            }
        }{
            \gray{\tcc{$\sigma_j$ is classified as a low-degree node; mark all $\sigma_j \in S_2$}}
            \For{$i \gets 1, \dots, t_2$} {
                \lIf{$S_2[i] = \sigma_j$} {
                    $I_L[i] \gets 1$
                }
            }
        }
    }
    \For{$i\gets 1, \dots, t_2$}{
        regenerate $R_{S_2[i]}$ using the same random seed\;
        \If{$I_L[i]=1$ and $R_{S_2[i]}/D[i]^p<K$}{
            $(K,r)\gets(R_{S_2[i]}/D[i]^p,S_2[i])$ \gray{\tcc{update the current minimum using a stored low-degree node}}
        }
    }
    \lIf{$K>\lambda$}{\Return Fail}
    \Return $r$\;
\end{algorithm}

\vspace{2mm}
\noindent{\bf The Ideas and Algorithm Descriptions.\ }
We can use the same strategy as $L_{\DI}$-sampling to develop a
three-pass algorithm for $L_{M_p}$-sampling. The idea is to partition the
nodes into a high-degree group $H$ and a low-degree group $L$, based on a
degree threshold of approximately $n^{1/(p+1)}$, using a witness set
$S_1$ sampled in the first pass. We then sample candidates from the two
groups in the second pass and compute their exact degrees in the third
pass. Rejection sampling can be applied within each group, after which
the two samples are merged according to the $M_p$-contributions of $H$
and $L$. As before, this approach needs three passes.

However, it is not clear how to compress this rejection-sampling approach
to two passes. We instead use the exponential-scaling approach. For every node
$j$, we draw an independent $R_j\sim\EXP(1)$ and scale it by $d_j^{-p}$.
The node minimizing $R_j / d_j^p$ is sampled with probability $d_j^p/M_p$.

The key observation is that we can store all low-degree nodes that {\em may}
attain the minimum in the first pass, without identifying the low-degree
nodes. Note that the minimum of $\{R_j / d_j^p\}_{j \in [n]}$ is distributed according to $\EXP(M_p)$, thus $\min\{R_j / d_j^p\}_{j \in [n]} \le \lambda$ with high probability.
Moreover, every node $j\in L$ has degree at most $O(n^{1/(p+1)})$. Therefore, if a low-degree node satisfies $R_j / d_j^p\leq\lambda$, then with high probability $R_j\leq O\left(n^{p/(p+1)} \lambda\right) = O(n^{-1/(p+1)} \log n)$. We can thus store every node satisfying the latter condition in $S_2$ during the first pass. In the second pass, we use $S_1$ to classify the nodes, compute the exact degrees of the stored low-degree candidates, and process every high-degree node using its estimated degree. Consequently, every node that may attain the minimum is processed by the algorithm. This gives a two-pass, sublinear-space algorithm.

We note that in Algorithm~\ref{alg:LMp-sampling}, by ``regenerate $R_j$'' we mean regenerating the corresponding random variables using Nisan's pseudorandom generator~\cite{Nisan90} with the same seed, which is stored in memory. Using Nisan's generator incurs only an additional $O(\log n)$ multiplicative factor in the space complexity and increases the error probability by at most $1/n$. We omit these overheads throughout the paper.

\vspace{2mm}
\noindent{\bf The Analysis.\ }
We now analyze Algorithm~\ref{alg:LMp-sampling} and prove the following theorem.

\begin{theorem}
\label{thm:LMp-upper}
Given an input data stream $\sigma$ of length $n$ and a similarity function
$f : U \times U \to [0, 1]$, for any $\eps \in (0,1)$ and $p > 0$, Algorithm~\ref{alg:LMp-sampling} is a $(1+\eps,o(1))$-$L_{M_p}$-sampler; it uses two passes and $O\left(\frac{1}{\eps^2}n^{1-1/(p+1)}\log n\right)$ words of space.
\end{theorem}

The pass and space complexities are clear from Algorithm~\ref{alg:LMp-sampling}.  In the remainder of this section, we show the correctness.

W.l.o.g., assume $\eps\ge n^{-1/(2(p+1))}$; otherwise, the stated space bound is at least $n$, in which case we can store all nodes and perform the sampling offline.

Recall that $Y_j$ is the total similarity weight between $S_1$ and $\sigma_j$.
Let $H=\{j\in[n]\mid Y_j\ge\theta\}$ and $L=\{j\in[n]\mid Y_j<\theta\}$, where $\theta$ is defined in Algorithm~\ref{alg:LMp-sampling}.

\begin{lemma}
    \label{lem:LMp-classification}
    With probability at least $1-O(n^{-2p-5})$, the following events hold: (1) $\min_{j\in H}d_j\ge n^{1/(p+1)}$, (2) $\max_{j\in L}d_j\le 4n^{1/(p+1)}$ and (3) for any $j\in H$, $\left(\frac{nY_j}{t_1}\right)^p=(1\pm \eps/8)d_j^p$.
\end{lemma}

\begin{proof}
    Since all nodes in $S_1$ are sampled uniformly at random from $[n]$, for any $j\in[n]$ with $d_j<n^{1/(p+1)}$, we have $\E[Y_j]=t_1\frac{d_j}{n} < \frac{\theta}{2}$. By a Chernoff bound,
    \begin{equation*}
        \Pr[Y_j\ge\theta] \le e^{-\theta/6} \le n^{-2(p+3)}.
    \end{equation*}
    Therefore, with probability at least $1-n^{-2p-5}$, every node of degree less
    than $n^{1/(p+1)}$ belongs to $L$. This proves the first item.
    
    On the other hand, for every $j$ with $d_j\ge n^{1/(p+1)}$, we have
    $\E[Y_j]\ge\theta/2$. By a Chernoff bound,
    \begin{equation*}
        \Pr\left[\abs{Y_j-\E[Y_j]}>\eps^{\prime} \E[Y_j] \right]
        \le 2\exp\left(-\frac{\E[Y_j] \cdot (\eps^{\prime})^2}{3}\right)
        \le 2\exp\left(-\frac{(\theta / 2) \cdot (\eps^{\prime})^2}{3}\right)
        = 2n^{-2(p+3)}.
    \end{equation*}
    Therefore, with probability at least $1-2n^{-2p-5}$, simultaneously for all
    $j$ with $d_j\ge n^{1/(p+1)}$, $\frac{nY_j}{t_1}=(1\pm\eps^{\prime}) d_j$.
    
    We condition on these two events in the rest of the proof, which hold with probability $1- O(n^{-2p-5})$.
    
    If $j\in L$, then either (1) $d_j<n^{1/(p+1)}$, or (2) $d_j\ge n^{1/(p+1)}$ and $Y_j<\theta$. In the second case,
    \begin{equation*}
        d_j \le \frac{nY_j}{(1-\eps^{\prime})t_1}
        < \frac{n\theta}{(1-\eps')t_1} \le 4n^{1/(p+1)}.
    \end{equation*}
    This proves the second item.
    
    Finally, for every $j\in H$, the first item implies
    $d_j\ge n^{1/(p+1)}$, and thus $nY_j/t_1=(1\pm\eps^{\prime})d_j$. Since $\eps^{\prime}=\eps/(16(p+1))$, we have $(nY_j/t_1)^p=(1\pm\eps/8)d_j^p$, proving the third item.
\end{proof}

The following property of exponential variables is widely used in weighted sampling.

\begin{lemma}[\cite{JW18}]
    \label{lem:exp-race}
    Let $w_1,\ldots,w_n>0$, let $R_1,\ldots,R_n$ be independent $\EXP(1)$ random variables, and set $K_i=R_i / w_i$. Let $I=\operatorname*{argmin}_{i\in[n]}K_i$ and $Z=\min_{i\in[n]}K_i$. Then, we have $\Pr[I=i]=\frac{w_i}{\sum_{j\in[n]}w_j}$, 
    $Z\sim\EXP\bigl(\sum_{j\in[n]}w_j\bigr)$, and $I$ is independent of
    $Z$.
\end{lemma}

Let $\cE_4$ denote the event that Lemma~\ref{lem:LMp-classification} holds. Then $\Pr[\cE_4] = 1 - O(n^{-2p-5})$. Define
\begin{equation*}
    \widehat d_j=
    \begin{cases}
        nY_j/t_1,&j\in H\\
        d_j,&j\in L
    \end{cases}
\end{equation*}
and $\widehat M_p=\sum_{j\in[n]}\widehat d_j^p$.
Let $J=\arg\min_{j\in[n]}\frac{R_j}{\widehat d_j^p}$ and $Z=\min_{j\in[n]}\frac{R_j}{\widehat d_j^p}$.

\begin{lemma}
    \label{lem:LMp-race}
    Conditioned on $\cE_4$, we have $\Pr[J=j \mid Z\leq\lambda]=(1\pm\eps/2)\frac{d_j^p}{M_p}$ for every $j\in[n]$. Moreover, $\Pr[Z>\lambda]\le n^{-5(p+3)}$.
\end{lemma}

\begin{proof}
    Fix any $S_1$ for which $\cE_4$ holds. Conditioned on
    this $S_1$, the values $\widehat d_1,\ldots,\widehat d_n$ are
    fixed, while $R_1,\ldots,R_n$ remain independent $\EXP(1)$ random
    variables. Applying Lemma~\ref{lem:exp-race} with $w_j=\widehat d_j^p$ and $\widehat M_p=\sum_{i\in[n]}\widehat d_i^p$, we obtain
    \begin{equation*}
        \Pr[J=j\mid Z\leq\lambda] = \frac{\widehat d_j^p}{\widehat M_p} \in \left[\frac{1-\eps/8}{1+\eps/8}, \frac{1+\eps/8}{1-\eps/8}\right]\frac{d_j^p}{M_p}=(1\pm\eps/2)\frac{d_j^p}{M_p}.
    \end{equation*}
    Moreover, since $M_p \ge n$,  $\Pr[Z>\lambda] = e^{-\widehat M_p\lambda} \le e^{-(1-\eps / 8) M_p\lambda} \le n^{-5(p+3)}$.
\end{proof}

\begin{lemma}
    \label{lem:LMp-candidates}
    The probability that Algorithm~\ref{alg:LMp-sampling} fails at line~\ref{ln:LMp-fail} is at most $e^{-t_2/6}$. Conditioned on $\cE_4$, if this event does not occur and $Z\le\lambda$, Algorithm~\ref{alg:LMp-sampling} returns $J$.
\end{lemma}

\begin{proof}
    Let $X$ be the number of nodes stored in $S_2$. Since
    $\Pr[R_j\le x]=1-e^{-x}\le x$ for every $x\ge0$, we have $\E[X] \le n \cdot (4n^{1/(p+1)})^p\lambda = \frac{t_2}{2}$. A Chernoff bound therefore gives $\Pr[X > t_2] \le e^{-t_2 / 6}$.
    
    Now condition on $\cE_4$ and suppose that $Z\le\lambda$. If $J\in H$, then Algorithm~\ref{alg:LMp-sampling} processes $J$ explicitly in the second pass.
    If $J\in L$, then Lemma~\ref{lem:LMp-classification} gives
    $d_J\le4n^{1/(p+1)}$, and hence $R_J=d_J^pZ \le(4n^{1/(p+1)})^p\lambda$.
    Thus, $J$ was stored in the first pass, and its exact degree is computed in the second pass. Therefore, in both cases Algorithm~\ref{alg:LMp-sampling} returns $J$.
\end{proof}

We now prove Theorem~\ref{thm:LMp-upper}. By Lemma~\ref{lem:LMp-classification}, Lemma~\ref{lem:LMp-race}, and Lemma~\ref{lem:LMp-candidates}, Algorithm~\ref{alg:LMp-sampling} fails with probability at most $O(n^{-2p-5})+n^{-5(p+3)}+e^{-t_2/6}=O(n^{-2p-5})$. Conditioned on $\cE_4$ and no failure, Algorithm~\ref{alg:LMp-sampling} outputs $\sigma_j$ with probability $(1 \pm \eps/2)\frac{d_j^p}{M_p}$. Consequently, for every $j\in[n]$,
\begin{equation*}
    \Pr[\text{output }j \mid \text{no failure}]
    =(1\pm\eps/2)\frac{d_j^p}{M_p}\pm O(n^{-2p-5})
    = (1\pm\eps)\frac{d_j^p}{M_p},
\end{equation*}
where the last equality holds since we have assumed that $\eps\ge n^{-1/(2(p+1))}$, and $d_j\le n$ and $M_p \le n^{p+1}$. 
Thus, Algorithm~\ref{alg:LMp-sampling} is a $(1+\eps,o(1))$-$L_{M_p}$-sampler.

\subsection{Removing the Assumption that $n$ Is Known}
\label{app:assumption}

We now explain how to obtain the random samples in the first pass in all of our algorithms without knowing the stream length. Note that after the first pass, $n$ is known automatically.

Let $t(N)$ denote the required sample size when the current stream length is $N$, and let $C>1$ be a sufficiently large constant. We assign each arriving node $i$ an independent rank $R_i$ chosen uniformly from $[0,1]$. After processing $N$ nodes, we set $q_N=\min\left\{1,\frac{C \cdot t(N)}{N}\right\}$
and store precisely the nodes $i\leq N$ satisfying $R_i\leq q_N$, discarding any previously stored nodes that no longer satisfy this condition when the threshold decreases.

For each sample size used in this paper, $t(N)$ is sublinear and $t(N)/N$ is decreasing. Hence, after each update, the stored set is exactly $\{i\leq N \mid R_i\leq q_N\}$. At the end of the first pass, if $q_n<1$, the number $X_n$ of stored nodes satisfies $X_n\sim\operatorname{Bin}(n,q_n)$ and $\E[X_n]=C \cdot t(n)$.
A Chernoff bound therefore shows that $X_n\geq t(n)$ with high probability. We then select the $t(n)$ stored nodes with the smallest ranks. These are the $t(n)$ smallest-ranked nodes in the entire stream and hence form a uniform sample without replacement.

Moreover, at every time $N\leq n$, the expected number $X_N$ of stored nodes is at most $C \cdot t(N)\leq C \cdot t(n)$. For a sufficiently large constant $K$, a Chernoff bound gives $\Pr\!\left[X_N>Kt(n)\right]\leq \exp(-\Omega(t(n)))$.
Note that the size of each sample used in our algorithms is $\Omega(\log n)$. A union bound over all $n$ updates therefore shows that $\max_{N\leq n}X_N=O(t(n))$ with high probability.

For Algorithms~\ref{alg:F0} and~\ref{alg:ell0}, we apply this construction only to the witness set $S_1$. For Algorithm~\ref{alg:Fp}, we apply this construction for both $S_1$ and $S_2$.  For Algorithm~\ref{alg:LMp-sampling}, we apply this construction to $S_1$. The standard concentration bounds for uniform sampling without replacement are no weaker than the corresponding bounds used in all these algorithms.  For $S_2$ in Algorithm~\ref{alg:LMp-sampling}, each node $i$ is assigned an independent exponential rank $R_i\sim\EXP(1)$ and is stored when $R_i\leq q_N$, where $q_N=O(N^{-1/(p+1)}\log N)$; since $q_N$ is decreasing and $\Pr[R_i\leq q_N]=1-e^{-q_N}\leq q_N$, the same Chernoff-bound argument preserves the stated space bound with high probability. 
Thus, the stated guarantees and space bounds are preserved without an additional pass.

\section{Lower Bounds}
\label{sec:lb}

In this section, we give several lower bounds to complement our algorithmic results. 

We use $\mathbf{e}_i^{k}$ to represent the $i$-th unit vector in $\mathbb{R}^{k}$, and $\mathbf{0}^{k}$ to represent the $k$-dimensional vector of all zeros.
Given $u \in \mathbb{R}^x$ and $v \in \mathbb{R}^y$, the notation $u \circ v$ stands for the vector in $\mathbb{R}^{x+y}$ obtained by concatenating $u$ and $v$.

Although we describe the hard instances using high-dimensional vectors, every stream item in our reductions has a succinct $O(\log n)$-bit representation using only $O(1)$ indices and constant-valued coordinates, from which the similarity function can be evaluated directly.

We prove lower bounds by performing reductions from the following well-studied communication problems.  

\begin{definition}[\INDEX]
    In the \INDEX\ problem, Alice holds an $n$-bit vector $X = (X_1, \ldots, X_n)$, and Bob holds an index $b \in [n]$.  The task is for Alice to send a message to Bob so that Bob can output $X_b$.
\end{definition}

\begin{lemma}[\cite{KN96}]
    \label{lem:INDEX}
    Any randomized algorithm that solves \INDEX\ with probability at least $0.51$ needs $\Omega(n)$ bits of communication. 
\end{lemma}

\begin{definition}[\DISJ]
    In the \DISJ\ problem, Alice holds an $n$-bit vector $X = (X_1, \ldots, X_n)$, and Bob holds an $n$-bit vector $Y = (Y_1, \ldots, Y_n)$.  The task is for Alice and Bob to exchange messages so that they output $1$ if there exists an index $i \in [n]$ such that $X_i = Y_i = 1$, and output $0$ otherwise.
\end{definition}

\begin{lemma}[\cite{BYJ+04}]
    \label{lem:DISJ}
    Any randomized algorithm that solves \DISJ\ with probability at least $0.51$ needs $\Omega(n)$ bits of communication. 
\end{lemma}

In the proofs of theorems in this section, we always use $I_{\mathtt{NO}}$ to denote an instance of \INDEX\ or \DISJ\ for which the output is $0$, and $I_{\mathtt{YES}}$ to denote an instance for which the output is $1$.

\subsection{One-Pass Streaming Algorithms}
We show the following lower bounds for one-pass streaming algorithms in the node-arrival model.

\begin{theorem}
    \label{thm:F0-lb}
    For any constant $C \ge 1$, any one-pass $(C, 0.49)$-approximation algorithm for computing $\rFo$ on a data stream of length $n$ needs $\Omega(n)$ bits of space.
\end{theorem}

\begin{proof}
    Set $t = 2C^2$ and $k = \frac{n}{2t} = \frac{n}{4C^2}$.  We perform a reduction from \INDEX\ to $\rFo$.  The input reduction is as follows: Alice, for each $i \in [k]$, adds $t$ points at location $0.6 \mathbf{e}_i^{k} \circ \mathbf{0}^{t k}$ if $X_i = 1$, or $t$ points at location $2\mathbf{e}_i^{k} \circ  \mathbf{0}^{t k}$ if $X_i = 0$.  Bob, for each $j \in [t k]$, adds a point at location $0.6 \mathbf{e}_b^{k} \circ 0.6 \mathbf{e}_j^{t k}$.  In total, Alice and Bob create $t k + t k = n$ points. We then define the similarity between two points $x$ and $y$ as follows: 
\begin{equation*}
f(x, y) \;=\;
\begin{cases}
    1, & \text{if } \lVert x-y \rVert_1 \leq 1, \\[6pt]
    0, & \text{otherwise.}
\end{cases}
\end{equation*}

If $X_b = 0$, then $\rFo(I_{\mathtt{NO}}) = (1+t) k = (1+2C^2) k$.  If $X_b = 1$, then 
\begin{equation*}
\rFo(I_{\mathtt{YES}})  = (k-1) + t \cdot \frac{1}{t + t k} + t k \cdot \frac{1}{t+1} < 2k.
\end{equation*}
We thus have
\begin{equation*}
\frac{\rFo(I_{\mathtt{NO}})}{C} > 2C k > C \cdot \rFo(I_{\mathtt{YES}}).
\end{equation*}
Therefore, any one-pass $C$-approximation streaming algorithm $\A$ for $\rFo$ using $s$ bits of space can be used to solve \INDEX\ with $s$ bits of communication: Alice simulates $\A$ on the set of points she has constructed, in any order, and sends the resulting memory configuration to Bob. Bob then continues the simulation on the set of points he has constructed, also in any order, and outputs $1$ if $\A$’s estimate of $\rFo$ is at most $2Ck$. By Lemma~\ref{lem:INDEX}, any such algorithm requires $\Omega(k) = \Omega(n)$ space.
\end{proof}

\begin{theorem}
    \label{thm:l0-lb}
    For any $\eps\in(0,0.9)$, any one-pass $(1+\eps,0.49)$-$L_{\DI}$-sampler on a data stream of length $n$ needs $\Omega(n)$ bits of space.
\end{theorem}

\begin{proof}
    We use the hard instance from the proof of Theorem~\ref{thm:F0-lb} with $t=400$ and $k=\frac{n}{2t} = \frac{n}{800}$. Let $V_A$ be the set of the $tk$ nodes added by Alice, and let $\cE_A$ be the event that the sampler outputs a node in $V_A$.

If $X_b=0$, then the total inverse-degree mass of $V_A$ is $k$ and $\DI(I_{\mathtt{NO}})=(t+1)k$, so the target probability of $\cE_A$ is $q_{\mathtt{NO}} = \frac{1}{t+1}$.

If $X_b=1$, then the total inverse-degree mass of $V_A$ is $(k-1) + \frac{1}{k+1} = \frac{k^2}{k+1}$ and $\DI(I_{\mathtt{YES}})=\frac{k^2}{k+1}+\frac{tk}{t+1}$.
Consequently, the target probability of $\cE_A$ is
\begin{equation*}
    q_{\mathtt{YES}} =\frac{k^2/(k+1)}{k^2/(k+1)+tk/(t+1)} \ge \frac{1}{3}.
\end{equation*}

Suppose that there is a one-pass
$(1+\eps,0.49)$-$L_{\DI}$-sampler $\A$ using $s$ bits of
space, where $\eps\in(0,0.9)$. Run $\ell=50$ independent copies of
$\A$ in parallel. Alice runs all copies on her portion of the
stream and sends their memory configurations to Bob. Bob then continues all copies on his portion of the stream. He outputs $1$ if at least one copy returns a node in $V_A$, and outputs $0$ otherwise. In particular, a failure is treated as not returning a node in $V_A$.

If $X_b=0$, then one copy returns a node in $V_A$ with probability at most $(1+\eps)q_{\mathtt{NO}} < \frac{1.9}{401}$.
Therefore, by a union bound, when $X_b=0$, the probability that the
protocol incorrectly outputs $1$ is at most $\frac{1.9 \ell}{401} < 0.24$.

If $X_b=1$, then one copy returns a node in $V_A$ with probability at
least
\begin{equation*}
(1-0.49)(1-\eps)q_{\mathtt{YES}} > 0.51\cdot 0.1\cdot\frac{1}{3} = \frac{17}{1000}.
\end{equation*}
Since copies are independent, the probability that none of them
returns a node in $V_A$ is at most $\left(1 - \frac{17}{1000}\right)^{\ell} < 0.43$.

The protocol therefore solves \INDEX\ with success probability greater
than $0.51$. Its total communication is $50s=O(s)$ bits. By
Lemma~\ref{lem:INDEX}, $s=\Omega(k)=\Omega(n)$.
\end{proof}

\begin{theorem}
    \label{thm:Mp-lb}
    For any $p>0$ and constant $C \geq 1$, any one-pass $(C, 0.49)$-approximation algorithm for computing $M_p$ on a data stream of length $n$ needs $\Omega(n)$ bits of space.
\end{theorem}

\begin{proof}
    Set $t = (2C^2)^{1/p}$ and $k = \frac{n}{2t} = \frac{n}{2(2C^2)^{1/p}}$.  We perform a reduction from \INDEX\ to $M_p$.  The input reduction is as follows: Alice, for each $i \in [k]$, adds $t$ points at locations $10 \mathbf{e}_i^{k} \circ \mathbf{e}_1^{t k} \circ 3, \ldots, 10 \mathbf{e}_i^{k} \circ \mathbf{e}_t^{t k} \circ 3$ if $X_i = 1$, or $t$ points at locations $-10 \mathbf{e}_i^{k} \circ \mathbf{e}_1^{t k} \circ 3, \ldots, -10 \mathbf{e}_i^{k} \circ \mathbf{e}_t^{t k} \circ 3$ if $X_i = 0$.  Bob, for each $j \in [t k]$, adds a point at location $10 \mathbf{e}_b^{k} \circ \mathbf{e}_{j}^{t k} \circ -3$.  In total, Alice and Bob create $t k + t k = n$ points. We then define the similarity between two points $x$ and $y$ as follows: 
\begin{equation*}
{f(x, y)} \;=\;
\begin{cases}
    1, & \text{if } x = y, \\[6pt]
    1, & \text{if } \lVert x-y \rVert_1 \in [4, 8], \\[6pt]
    0, & \text{otherwise.}
\end{cases}
\end{equation*}

If $X_b = 0$, then $M_p(I_{\mathtt{NO}}) = n$. If $X_b = 1$, then 
\begin{equation*}
M_p(I_{\mathtt{YES}}) = t(k-1)+t(tk+1)^p+tk(t+1)^p > tk\cdot t^p = C^2n.
\end{equation*}
Thus, for any constant $C \geq 1$, we have
\begin{equation*}
\frac{M_p(I_{\mathtt{YES}})}{C} > C \cdot M_p(I_{\mathtt{NO}}).
\end{equation*}

Since any one-pass $C$-approximation streaming algorithm for $M_p$ using $s$ bits of space can be used to solve \INDEX\ with $s$ bits of communication, by Lemma~\ref{lem:INDEX}, any such algorithm requires $\Omega(k) = \Omega(n)$ space.
\end{proof}

\begin{theorem}
\label{thm:LMp-lb}
For any $p>0$ and $\eps\in(0,0.9)$, any one-pass
$(1+\eps,0.49)$-$L_{M_p}$-sampler on a data stream of length $n$
needs $\Omega(n)$ bits of space.
\end{theorem}

\begin{proof}
    We use the same hard instance from the proof of
Theorem~\ref{thm:Mp-lb}. Choose a constant $t$ satisfying $t^p\geq 1000$, and set $k=\frac{n}{2t}$.

Let $A_{\bar b}$ be the set of the $t(k-1)$ nodes added by Alice for
coordinates $i\neq b$, and let $\cE_{\bar b}$ be the event that the
sampler outputs a node in $A_{\bar b}$.

If $X_b=0$, then every node has degree $1$. Therefore, the target
probability of $\cE_{\bar b}$ is
\begin{equation*}
    q_{\mathtt{NO}}=\frac{t(k-1)}{n}=\frac{k-1}{2k}\geq \frac{1}{3},
\end{equation*}
where the last inequality holds for $k\geq3$.

If $X_b=1$, then the $t$ nodes added by Alice for coordinate $b$ have
degree $tk+1$, the $tk$ nodes added by Bob have degree $t+1$, and all
nodes in $A_{\bar b}$ have degree $1$. Hence the target probability of $\cE_{\bar b}$ is
\begin{equation*}
    q_{\mathtt{YES}} = \frac{t(k-1)}{t(k-1)+t(tk+1)^p+tk(t+1)^p} 
    \le \frac{t(k-1)}{tk(t+1)^p} < \frac{1}{(t+1)^p} < \frac{1}{1000}.
\end{equation*}

Suppose that there is a one-pass
$(1+\eps,0.49)$-$L_{M_p}$-sampler $\A$ using $s$ bits of space, where
$\eps\in(0,0.9)$. Run $\ell=50$ independent copies of $\A$ in parallel.
Alice runs all copies on her portion of the stream and sends their memory configurations to Bob. Bob then continues all copies on his portion of the stream.
He outputs $0$ if at least one copy returns a node in $A_{\bar b}$, and
outputs $1$ otherwise. A failure is treated as not returning a node in
$A_{\bar b}$.

If $X_b=0$, then one copy returns a node in $A_{\bar b}$ with probability at least
\begin{equation*}
    (1-0.49)(1-\eps)q_{\mathtt{NO}} > 0.51\cdot0.1\cdot\frac{1}{3} = \frac{17}{1000}.
\end{equation*}
Since the copies are independent, the probability that none of them
returns a node in $A_{\bar b}$ is at most $\left(1-\frac{17}{1000}\right)^\ell \le \exp\left(-\frac{17}{20}\right) < 0.43$.
Thus, when $X_b=0$, the protocol outputs $0$ with probability greater
than $0.57$.

If $X_b=1$, then one copy returns a node in $A_{\bar b}$ with probability at most
\begin{equation*}
    (1+\eps)q_{\mathtt{YES}} < \frac{1.9}{1000}.
\end{equation*}
Therefore, by a union bound, the probability that the protocol incorrectly outputs $0$ is at most $\ell(1+\eps)q_{\mathtt{YES}} < 0.1$.

The protocol therefore solves \INDEX\ with success probability greater
than $0.51$. Its total communication is $\ell s=O(s)$ bits. By
Lemma~\ref{lem:INDEX}, $s=\Omega(k)=\Omega(n)$.
\end{proof}

\subsection{Multi-Pass Streaming Algorithms}
We give the following multi-pass lower bound.

\begin{theorem}
    \label{thm:F0-lb-multi}
    For any $\eps \in (0, 0.9)$, any $O(1)$-pass $(1+\eps, 0.49)$-approximation algorithm for computing $\rFo$ on a data stream of length $n$ needs $\Omega\left(\min\left(\sqrt{\frac{n}{\eps}}, n \right)\right)$ bits of space.
\end{theorem}

\begin{proof}
    Set $t = \max(1, 5 \sqrt{\eps n})$ and $k = \frac{n}{2t} = \min \left(\frac{n}{2}, \frac{1}{10} \sqrt{\frac{n}{\eps}}\right)$.
We perform a reduction from \DISJ\ to $\rFo$.  The input reduction is as follows:  Alice, for each $i \in [k]$,  adds $t$ points at locations $0.6 \mathbf{e}_1^t \circ 4i, \ldots, 0.6 \mathbf{e}_t^t \circ 4i$ if $X_i = 0$, or $t$ points at locations $0.6 \mathbf{e}_1^t \circ (4i-0.4), \ldots, 0.6 \mathbf{e}_t^t \circ (4i-0.4)$ if $X_i = 1$.  Bob, for each $i \in [k]$,  adds $t$ points all at location $\mathbf{0}^t \circ 4i$ if $Y_i = 0$, or $t$ points all at location $\mathbf{0}^t \circ (4i+0.4)$ if $Y_i = 1$.  In total, Alice and Bob create $2t k = n$ points. We then define the similarity between two points $x$ and $y$ as follows: 
\begin{equation*}
f(x, y) \;=\;
\begin{cases}
    1, & \text{if } \lVert x-y \rVert_1 \leq 1, \\[6pt]
    0, & \text{otherwise.}
\end{cases}
\end{equation*}

If \DISJ$(X, Y) = 0$, then
\begin{equation*}
\rFo(I_{\mathtt{NO}}) = k \cdot \left(\frac{t}{t+1} + \frac{t}{2t}\right) \le \frac{3k}{2} = \frac{3n}{4t}.
\end{equation*}
If \DISJ$(X, Y) = 1$, then there exists $r \ge 1$ coordinates satisfying $X_i = Y_i = 1$ and
\begin{equation*}
\rFo(I_{\mathtt{YES}}) = (k-r) \left(\frac{t}{t+1} + \frac{t}{2t}\right) + r(t+1) \ge \rFo(I_{\mathtt{NO}}) + \left(t - \frac{1}{2} + \frac{1}{t+1}\right) \ge \rFo(I_{\mathtt{NO}}) + \frac{3t}{4}. 
\end{equation*}
It follows that
\begin{equation*}
    \rFo(I_{\mathtt{YES}}) \ge \left(1 + \frac{t^2}{n}\right) \rFo(I_{\mathtt{NO}}) \ge (1 + 25 \eps) \rFo(I_{\mathtt{NO}}).
\end{equation*}
Thus, for any $\eps \in (0, 0.9)$, easy calculations give
\begin{equation*}
(1-\eps) \cdot \rFo(I_{\mathtt{YES}}) > (1+\eps) \cdot \rFo(I_{\mathtt{NO}}).
\end{equation*}
Therefore, any $O(1)$-pass $(1+\eps)$-approximation streaming algorithm for $\rFo$ using $s$ bits of space can be used to solve \DISJ\ with $O(s)$ bits of communication. By Lemma~\ref{lem:DISJ}, any such algorithm requires $\Omega(k) = \Omega\left(\min \left(n, \sqrt{\frac{n}{\eps}}\right)\right)$ space.
\end{proof}

\begin{theorem}
    \label{thm:l0-lb-multi}
    For any $\eps\in(0,0.9)$, any $O(1)$-pass $(1+\eps,0.49)$-$L_{\DI}$-sampler on a data stream of length $n$ needs $\Omega(\sqrt{n})$ bits of space.
\end{theorem}

\begin{proof}
    We use the hard instance from the proof of Theorem~\ref{thm:F0-lb-multi}, but choose $t=4k$, so the stream contains $n=2tk=8k^2$ nodes.

Let $A_{\cap}$ be the set of Alice's nodes belonging to coordinates $i$ for which $X_i=Y_i=1$, and let $\cE_{\cap}$ be the event that the sampler outputs a node in $A_{\cap}$. If \DISJ$(X,Y)=0$, then $A_{\cap}=\emptyset$ and $\cE_{\cap}$ never occurs.

If \DISJ$(X,Y)=1$, then there exists $r \ge 1$ coordinates satisfying $X_i = Y_i = 1$ and the total inverse-degree mass of $A_{\cap}$ is $rt=4rk$, while
\begin{equation*}
\DI(I_{\mathtt{YES}})=(k-r) \left(\frac{t}{t+1} + \frac{t}{2t}\right) + r(t+1) \le \frac{11rk}{2}.
\end{equation*}
Hence the probability of $\cE_{\cap}$ is $q_{\cap} = \frac{rt}{\DI(I_{\mathtt{YES}})} \ge \frac{8}{11}$.

Suppose that there is an $O(1)$-pass $(1+\eps,0.49)$-$L_{\DI}$-sampler $\A$ using $s$ bits of space, where $\eps\in(0,0.9)$. Run $\ell=50$ independent copies of $\mathcal A$ in parallel. Alice and Bob simulate these copies by exchanging their memory configurations after their respective portions of each pass. After the samplers terminate, they use $O(\ell\log k)=O(\log k)$ additional bits to determine whether any
copy returned a node in $A_{\cap}$. They output $1$ if and only if
this occurs. A failure is treated as not returning a node in
$A_{\cap}$.

If \DISJ$(X,Y)=0$, then $A_{\cap}=\emptyset$, so the protocol always
outputs $0$ and is therefore always correct.

If \DISJ$(X,Y)=1$, then one copy returns a node in $A_{\cap}$ with probability at least
\begin{equation*}
(1-0.49)(1-\eps)q_{\cap} \ge 0.51\cdot 0.1\cdot\frac{8}{11} = \frac{51}{1375}.
\end{equation*}
Since copies are independent, the probability that none of them
returns a node in $A_{\cap}$ is at most $\left(1 - \frac{51}{1375}\right)^{\ell} < e^{-1} < 0.37$.

The protocol therefore solves \DISJ\ with success probability greater
than $0.51$. The total communication is $O(\ell s + \ell \log k) = O(s + \log k)$ bits. By Lemma~\ref{lem:DISJ}, $s=\Omega(k)=\Omega(\sqrt{n})$.
\end{proof}

\begin{theorem}
    \label{thm:Fp-lb-multi}
    For any $p > 0$ and $\eps \in (0, 0.9)$, any $O(1)$-pass $(1+\eps, 0.49)$-approximation algorithm for computing $\rFp$ on a data stream of length $n$ needs $\Omega\left(\min \left(\eps^{-1/(p+1)}n^{1-1/(p+1)}, n \right)\right)$ bits of space.
\end{theorem}

\begin{proof}
    Let $c_p = \min(1, 2^p -1) > 0$. Set 
\begin{equation*}
t = \max\left(1, \left(\frac{10 \eps n}{c_p}\right)^{1/(p+1)}\right) \text{ and } 
\end{equation*}
\begin{equation*}
k = \frac{n}{2 t} = \min \left( \frac{n}{2}, \frac{c_p^{1/(p+1)}}{2 \cdot 10^{1/(p+1)}}\eps^{-1/(p+1)} n^{1-1/(p+1)}\right).
\end{equation*}
We perform a reduction from \DISJ\ to $M_p$. The input reduction is as follows: Alice, for each $i \in [k]$,  adds $t$ points at locations $10 \mathbf{e}_i^k \circ \mathbf{e}_1^t \circ 3, \ldots, 10 \mathbf{e}_i^k \circ \mathbf{e}_t^t \circ 3$ if $X_i = 1$, or $t$ points at locations $10 \mathbf{e}_i^k \circ \mathbf{e}_1^t \circ 30, \ldots, 10 \mathbf{e}_i^k \circ \mathbf{e}_t^t \circ 30$ if $X_i = 0$. Bob, for each $i \in [k]$, adds $t$ points at locations $10 \mathbf{e}_i^k \circ \mathbf{e}_1^t \circ -3, \ldots, 10 \mathbf{e}_i^k \circ \mathbf{e}_t^t \circ -3$ if $Y_i = 1$, or $t$ points at locations $10 \mathbf{e}_i^k \circ \mathbf{e}_1^t \circ -30, \ldots, 10 \mathbf{e}_i^k \circ \mathbf{e}_t^t \circ -30$ if $Y_i = 0$. In total, Alice and Bob create $2kt = n$ points. We then define the similarity between two points $x$ and $y$ as follows: 
\begin{equation*}
f(x, y) \;=\;
\begin{cases}
    1, & \text{if } x = y, \\[6pt]
    1, & \text{if } \lVert x-y \rVert_1 \in [4, 8], \\[6pt]
    0, & \text{otherwise.}
\end{cases}
\end{equation*}

If \DISJ$(X,Y) = 0$, then $M_p(I_{\mathtt{NO}}) = n$. If \DISJ$(X,Y) = 1$, then there exists $r \ge 1$ coordinates satisfying $X_i = Y_i = 1$ and
\begin{equation*}
M_p(I_{\mathtt{YES}}) = (n-2rt) + 2rt \cdot (t+1)^p \geq n + 2 c_p t^{p+1} \ge (1+20\eps)n,
\end{equation*}
where we apply the inequality $(t+1)^p - 1 \ge c_p t^p$ for any $t \ge 1$ and $p > 0$.
Thus, for any $\eps \in (0,0.9)$, easy calculations give
\begin{equation*}
(1-\eps)\cdot M_p(I_{\mathtt{YES}}) > (1+\eps) \cdot M_p(I_{\mathtt{NO}}).
\end{equation*}
Therefore, any $O(1)$-pass $(1+\eps)$-approximation streaming algorithm for $M_p$ using $s$ bits of space can be used to solve \DISJ\ with $O(s)$ bits of communication. By Lemma~\ref{lem:DISJ}, any such algorithm requires $\Omega(k) = \Omega(\min(n, \eps^{-1/(p+1)} n^{1-1/(p+1)}))$ space.
\end{proof}

\begin{theorem}
\label{thm:LMp-lb-multi}
For any $p>0$ and $\eps\in(0,0.9)$, any $O(1)$-pass
$(1+\eps,0.49)$-$L_{M_p}$-sampler on a data stream of length $n$
needs $\Omega\left(n^{1-1/(p+1)}\right)$ bits of space.
\end{theorem}

\begin{proof}
    We use the same hard instance from the proof of
Theorem~\ref{thm:Fp-lb-multi}, but choose $t=(4k)^{1/p}$. The resulting stream contains $n=2tk = \Theta\left(k^{1+1/p}\right)$ nodes.

Let $A_{\cap}$ be the set of Alice's nodes belonging to coordinates
$i$ for which $X_i=Y_i=1$, and let $\cE_{\cap}$ be the event that the
sampler outputs a node in $A_{\cap}$. If \DISJ$(X,Y)=0$, then
$A_{\cap}=\emptyset$, and hence $\cE_{\cap}$ never occurs.

Suppose that \DISJ$(X,Y)=1$, and let $r\geq1$ be the number of coordinates satisfying $X_i=Y_i=1$. For every such coordinate, the corresponding $t$ Alice nodes and $t$ Bob nodes have degree $t+1$, while all other nodes have degree $1$. Therefore, the total mass of $A_{\cap}$ is $rt(t+1)^p$, and $M_p(I_{\mathtt{YES}}) = n-2rt+2rt(t+1)^p$.
Hence the target probability of $\cE_{\cap}$ is
\begin{equation*}
    q_{\cap} = \frac{rt(t+1)^p}{n-2rt+2rt(t+1)^p}
    \ge \frac{(t+1)^p}{2k+2(t+1)^p} \ge \frac{2}{5}.
\end{equation*}

Suppose that there is an $O(1)$-pass
$(1+\eps,0.49)$-$L_{M_p}$-sampler $\A$ using $s$ bits of space, where
$\eps\in(0,0.9)$. Run $\ell=50$ independent copies of $\A$ in parallel.
Alice and Bob simulate these copies by exchanging their memory configurations after their respective portions of each pass. After the samplers terminate, they use $O(\ell\log k)=O(\log k)$ additional bits to determine whether any copy returned a node in $A_{\cap}$. They output $1$ if and only if this occurs.
A failure is treated as not returning a node in $A_{\cap}$.

If \DISJ$(X,Y)=0$, then $A_{\cap}=\emptyset$, so the protocol always
outputs $0$ and is therefore always correct.

If \DISJ$(X,Y)=1$, then one copy returns a node in $A_{\cap}$ with
probability at least
\begin{equation*}
    (1-0.49)(1-\eps)q_{\cap} > 0.51\cdot0.1\cdot\frac{2}{5} = \frac{51}{2500}.
\end{equation*}
Since copies are independent, the probability that none of them returns
a node in $A_{\cap}$ is at most $\left(1-\frac{51}{2500}\right)^\ell < e^{-1} <0.37$.

The protocol therefore solves \DISJ\ with success probability greater than $0.51$. The total communication is $O(\ell s+\ell\log k)=O(s+\log k)$ bits. By Lemma~\ref{lem:DISJ}, $s=\Omega(k) = \Omega(n^{1-1/(p+1)})$.
\end{proof}

\section{Concluding Remarks}
\label{sec:conclude}

In this work, we studied a set of statistical problems defined on similarity graphs in the node-arrival data stream model. 
Several questions remain open. For example, while our upper and lower bounds match in terms of $n$, they still differ with respect to $\eps$; can these bounds be made tight?
It would also be interesting to study other statistical problems defined on similarity graphs in the node-arrival streams.

\bibliography{paper}
\bibliographystyle{plain}

\end{document}